\documentclass[review]{elsarticle}

\usepackage{lineno,hyperref}
\modulolinenumbers[5]
\usepackage{float}
\usepackage{graphicx}
\usepackage{subfig}
\usepackage{mathrsfs}
\usepackage{graphics}
\usepackage{float}
\usepackage{amsmath}
\usepackage{amsfonts}
\usepackage{optidef}
\usepackage[subnum]{cases}

\usepackage{optidef}
\usepackage{graphicx}
\usepackage{float}
\usepackage{graphicx}
\usepackage{subfig}
\usepackage{mathrsfs}
\usepackage{float}
\usepackage{amsmath}
\usepackage{amsfonts}
\usepackage{booktabs}
\usepackage{algorithm}
\usepackage[noend]{algpseudocode}
\usepackage{multirow}
\algdef{SE}[DOWHILE]{Do}{doWhile}{\algorithmicdo}[1]{\algorithmicwhile\ #1}

\usepackage[letterpaper, portrait, margin=1in]{geometry}

\usepackage{optidef}

\newtheorem{theorem}{Theorem}
\usepackage{amsthm}
\newtheorem{prop}{Proposition}
\theoremstyle{definition}

\newtheorem{lemma}[theorem]{Lemma}

\makeatletter
\newcommand*{\centerfloat}{%
  \parindent \z@
  \leftskip \z@ \@plus 1fil \@minus \marginparwidth
  \rightskip \leftskip
  \parfillskip \z@skip}
\makeatother

\usepackage{setspace}
\usepackage{natbib}
\usepackage{algorithm}
\usepackage[noend]{algpseudocode}
\algdef{SE}[DOWHILE]{Do}{doWhile}{\algorithmicdo}[1]{\algorithmicwhile\ #1}

\usepackage{appendix}

\usepackage{array}
\usepackage{threeparttable}
\usepackage{multirow}
\usepackage{rotating}
\usepackage{makecell}

\journal{}

\biboptions{authoryear}

\begin{document}
\begin{frontmatter}

\title{Assignment-based Path Choice Estimation for Metro Systems Using Smart Card Data}


\author[label1,label4]{Baichuan Mo}
\author[label2]{Zhenliang Ma\corref{mycorrespondingauthor}}
\author[label3]{Haris N. Koutsopoulos}
\author[label4]{Jinhua Zhao}
\address[label1]{Department of Civil and Environmental Engineering, Massachusetts Institute of Technology, Cambridge, MA 02139}
\address[label2]{Department of Civil and Environmental Engineering, Monash University, }
\address[label3]{Department of Civil and Environmental Engineering, Northeastern University, Boston, MA 02115}
\address[label4]{Department of Urban Studies and Planning, Massachusetts Institute of Technology, Cambridge, MA 20139}
\cortext[mycorrespondingauthor]{Corresponding author}

\begin{abstract}

Urban rail services are the principal means of public transportation in many cities. To understand the crowding patterns and develop efficient operation strategies in the system, obtaining path choices is important. This paper proposed an assignment-based path choice estimation framework using automated fare collection (AFC) data. The framework captures the inherent correlation of crowding among stations, as well as the interaction between path choice and left behind. The path choice estimation is formulated as an optimization problem. The original problem is intractable because of a non-analytical constraint and a non-linear equation constraint. A solution procedure is proposed to decompose the original problem into three tractable sub-problems, which can be solved efficiently. The model is validated using both synthetic data and real-world AFC data in Hong Kong Mass Transit Railway (MTR) system. The synthetic data test validates the model's effectiveness in estimating path choice parameters, which can outperform the purely simulation-based optimization methods in both accuracy and efficiency. The test results using actual data show that the estimated path shares are more reasonable than survey-derived path shares and uniform path shares. Model robustness in terms of different initial values and different case study dates are also verified. 

\end{abstract}


\end{frontmatter}

% \linenumbers

\section{Introduction}\label{intro}

With the increasing of city scale and population, metro systems are playing more and more important roles in urban transportation. Understanding the passenger flow distribution in the metro system is crucial for transit agencies to adjust operation strategies and better accommodate passengers. Simulation and transit assignment models are powerful instruments to obtain the passenger flows in the network, which can help to monitor and evaluate the system performance. Two important inputs are required for these kinds of models: the origin-destination (OD) demand matrix and passengers' path choice behavior. Thanks to the widely adopted automated fare collection (AFC) system, the station-to-station OD matrix in a metro network is usually known, while the path choice is hardly to be observed directly. Therefore, estimating the path choice becomes an important premise for system performance monitoring.

On-site survey is a conventional way to estimate path choices. However, the survey-based method has always been criticized due to the time-consuming and labor-intensive process. In addition, given the changes of metro rail network and operation timetables, the survey results may be out of date. To overcome these disadvantages, researchers have proposed path choice estimation methods using AFC data.

The AFC system are designed to conveniently charge passengers who use the metro system. When passengers tap in or tap out in the system with a smart card, the exact locations and time of the transactions will be recorded, which provide rich information for analyzing passenger behaviors. In the context of path choice estimation, the AFC data-based methods can be categorized into two groups: path-identification methods \citep{kusakabe2010estimation,zhou2012model,kumar2018robust} and parameters-inference methods \citep{sun2012rail,sun2015integrated,zhao2017estimation,xu2018learning}. The former studies aimed to identify the exact path chosen by a user. The path attributes (e.g. walking time, in-vehicle time) are used to evaluate how likely paths are chosen by passengers. While the later studies formulated probabilistic models to describe the random process of passengers' path choice behaviors. Bayesian inference is usually used to estimate the corresponding choice parameters, based on which path choice fractions can be derived. Despite using different methods, the key idea for those AFC data-based studies are similar. They all attempted to match the model-derived journey time with the observed journey time from AFC data. Since model-derived journey time is determined by choice parameters, the observed journey time thus provide a source to calibrate path choice behaviors. However, this type of methods may fail when the left behind (denied boarding) phenomenons are not well addressed. 

In a congested metro system, passengers are likely to be left behind due to the limited capacity of trains. Left behind will cause passengers' waiting time on the platform to increase, thus increase their total travel time. It may happen that the journey time for a long-distance route without left behind is close to that of a short-distance route with left behind, which makes the two routes indistinguishable using the purely journey time based methods \citep{zhu2017passenger}. Several studies have taken left behind into consideration explicitly or implicitly. For example, \cite{sun2015integrated} considered the delay caused by left behind as the part of travel time variability. This method is unable to distinguish the choice of routes with very similar journey time distribution. \cite{sun2012rail} and \cite{zhao2017estimation} assumed the left-behind probability for different stations are independent, and explicitly estimated the left behind probability before inferring the path choice fraction. These partially addressed the left behind problem. However, the independent left behind assumption neglects the important interactions in the network. In the real world, the left behind is caused by the interaction between supply and demand. A station with high entry demand may cause the adjacent stations to be congested because the remaining capacity for next station will become very limited. Therefore, the left behind probability for different stations are not independent. Moreover, it is not reasonable to consider path choice and left behind separately. These two components are interacted, where one could easily change another and vice-versa. Thus, the path choice estimation model needs to consider the correlation of left behinds among platforms, as well as the interactions between path choice and left behind. One of the alternatives is embedding the transit assignment model into path choice estimation with the information of network topology and train operation schedule. However, as known in the literature, the schedule-based dynamic transit assignment is a complicated problem with no direct closed form \citep{song2017statistical}. None of the aforementioned studies have incorporated it into path choice estimation, which, as a result, neglected the important network operation and interaction information.

This paper proposed a new path choice estimation framework, which incorporates the network topology and train operation information by embedding a transit assignment model. It can capture the left behinds correlation among platforms, and address the interactions between path choice and left behind. We formulate the path choice estimation as an optimization problem. The original problem is intractable due to non-analytical and non-linear equation constraints. We proposed a solution procedure to decompose the original problem into three tractable sub-problems, where each of them can be solved efficiently. The model is validated using data from Hong Kong Mass Transit Railway (MTR) system, which affirms the effectiveness of the proposed model in path choice estimation.

This remainder of this paper is organized as follows: Section \ref{lit} reviewed the related studies in the literature. Section \ref{method} describes the modeling framework, including network representation, problem definition, and solution procedures. The model was validated using both synthetic data and actual data from Hong Kong MTR in Section \ref{case_study}. Model robustness was also tested. Main ﬁndings and future research directions are summarized in Section \ref{discuss}.

\section{Literature Review}\label{lit}
Considerable literature exists on rail transit path choice estimation. On-site survey is a conventional way to estimate path choices. Preference data, such as stated preference (SP) and revealed preference (RP), are collected and analyzed by demand modeling methods. For example, \cite{lam2002transit} applied a path-size logit model to estimate route choice behaviors in Singapore with mixed SP and RP data. \cite{nazem2011demographic} adopted a discrete choice model to estimate passengers' route choice behaviors for different demographic groups. The household travel survey in Canada was used. \cite{eluru2012travel} used a mixed logit framework to study transit route choices in Montreal, Canada. A Google Map-based RP survey was deployed to collect the data. A methodological review on survey-based route choice estimation can be found in \cite{prato2009route}.

Recently, the emergence of smart card has shifted the research filed toward data-driven path choice estimation using historical transactions, rather than collecting route choice information with physical surveys. These studies can be categorized into two categories: path-identification methods and parameters-inference methods. In terms of path identification, \cite{kusakabe2010estimation} proposed a algorithm to identify the exact train that a passenger boarded using smart card data, which then gave the results of path choice. Based on the case study in Japan, the model were implicitly verified with the load weight of trains and GPS logs from a probe person survey. \cite{zhou2012model} proposed a path identification method based on the maximum
likelihood boarding plan, which assumes each individual will choose the path with the highest matching degree. The actual passenger data from the Beijing subway system were used as a case study. \cite{kumar2018robust} proposed a trip chaining method to infer the most likely trajectory of transit passengers using AFC and General Transit Feed
Speciﬁcation (GTFS) data. The method is applied to transit data from the Twin Cities and implicitly verified by the automatic passenger count data. The disadvantages of path-identification methods are as follows. First, these methods are usually applied at individual level, which may bring great computational challenges in large-scale and high-demand networks. Second, for the purpose of service quality evaluation, operators care more about the network-level path choice. Path-identification methods can only obtain the network-level path choice by aggregating the individual-level behaviors, which may induce estimation errors for some OD pairs with limited sample size. 

In contrast, parameters-inference methods can direct output the network-level path choice, which is more suitable for the system performance evaluation. Parameters-inference methods usually connect path shares with path attributes by constructing behavioral models (e.g. discrete choice model), and estimate the corresponding parameters in the constructed models. 
\cite{sun2012rail} proposed a probabilistic model for path choice estimation using AFC data. They first estimated platform elapsed time for transfer stations and through stations, then proposed a Gaussian mixture model to describe path choice fractions based on journey time distribution. The model is quantitatively validated with a simple synthetic data set and qualitatively validated based on Beijing metro systems. \cite{sun2015integrated} proposed an integrated Bayesian approach to estimate the network-level path choices. The path choice is described by a multinomial logit model with parameters to be estimated. The model is qualitatively validated based on the Singapore metro system.
\cite{zhao2017estimation} proposed a probabilistic model to estimated the route choice patterns using AFC data. They first estimated the number of trains waited by passengers, which is equivalent to the left behind rate. Then the path choice fractions were modeled and estimated based on a Gaussian mixture model.  \cite{xu2018learning} proposed a Bayesian inference approach to estimate the path choices parameters in logit model using AFC data. The Metropolis-Hasting sampling is used to calibrate the model parameters.  

As we mentioned in Section \ref{intro}, left-behind phenomenon is important to estimate the network-level path choice. However, few studies have well addressed this problem. The purely journey time based methods \citep{sun2012rail,sun2015integrated,xu2018learning} considered the waiting time caused by left behind as part of total journey time, which cannot distinguish long-distance paths without left behind and short-distance paths with left behind when they have very similar total journey time. \cite{zhao2017estimation} assumed left behinds are independent across stations, and considered the left behind and path choice separately, which neglected the interaction between supply and demand in the network.  Thus, it remains a challenge to develop a comprehensive path choice estimation framework which can capture the left behinds correlation among platforms, and address the interactions between path choice and left behind.

\section{Methodology}\label{method}
\vspace{3pt}
\subsection{Network Representation}\label{network_repre}
To capture the network interaction and operation information, dynamic transit assignment module should be incorporated into the model. For clarification, the term "assignment" in this paper represents the network loading process \citep{song2017statistical}, in which the passengers' route choices are known and treated as input. A typical way to represent transit network for assignment is using the Time-space (TS) hyper-network \citep{nguyen2001modeling,hamdouch2008schedule,hamdouch2011schedule}, where one station in the metro system are expanded into a series of nodes, representing the station at different time intervals. The length of the time interval $\tau$ is usually set as the minimal headway. For example, a station $a$ in the metro system will be expanded to nodes series ($a_1, a_2, ..., a_N$), where $a_1$ represents station $a$ at time 7:00-7:02; $a_2$ represents station $a$ at time interval 7:02-7:04, etc. Apparently, this fine-grained method may not be practical for the real-world application since the TS network can be extremely large. Consider a metro system with 100 stations and minimal headway of 2 minutes (e.g. the MTR network in our case study). To perform a 2-hours assignment, one station will be expanded to 60 TS nodes. The total number OD pairs in this TS network is approximately 36 millions, which brings huge computational challenges. However, the path choice calibration problem actually does not require such a fine-grained framework. It is commonly assumed that path choice behaviors are static for a specific time period (e.g. one hour). Therefore, a more aggregated network representation is needed.

Let us consider a studying time period $T$ and divide it into $n$ elementary time intervals of length $\tau$. Different to the typical TS network where $\tau$ is equal to the minimal headway, we set $\tau$ as a larger time intervals which includes several headways (e.g. $\tau = 15$ minutes). In the aggregated TS network, many headway-level behaviors, such as left behinds, cannot be explicitly modelled. But the trade-off is that we can obtain a more spares TS network which can be applied to large scale metro systems.

Consider two stations $i$ and $j$ in a metro system with different routes between them. The route set is denoted as $\mathscr{R}(i,j)$. Our purpose is to calculate the choice proportions of these routes for different time intervals. Now we expand the $i$ and $j$ into a sequence of TS nodes with time interval $\tau$, representing as $(i_1,...,i_m,...,i_N)$ and $(j_1,...,j_n,...,j_N)$, where $N = T/\tau$; $i_m$ represents station $i$ at time interval $m$; $j_n$ represents station $j$ at time interval $n$. Based on these notations, we define the following variables.

\begin{itemize}
\item \textbf{OD entry flow} (denoted as $q^{i_m,j}$): Number of people with origin $i$ and destination $j$ and entering station $i$ within time interval $m$. The OD entry flow is the demand input for transit assignment model. It can be obtained from the AFC data directly. The set of all $q^{i_m,j}$ is denoted as $\boldsymbol{q_e}$.

\item \textbf{OD entry-exit flow} (denoted as $q^{i_m,j_n}$): Number of people who enter station $i$ within time interval $m$ and exit at station $j$ within time interval $n$ ($m \leq n$). $q^{i_m,j_n}$ can be obtained from the transit assignment model, which contains the information of when passengers exit the system. Importantly, the ground truth OD entry-exit flow is available in the AFC data, which provides us opportunities to calibrate the path choice. In this study, OD entry-exit flow is similar to previous research which used journey time as the ground truth information.

\item \textbf{Path choice fraction (or Path share)} (denoted as $p_r^{i_m,j}$): The probability of path $r$ being chosen within time interval $m$, where $r \in \mathscr{R}(i,j)$. The subscript $m$ incorporates the dynamic (time-dependent) choice behavior. By definition, we have $ 0 \leq p_r^{i_m,j} \leq 1$ and $\sum_{r \mathscr{R}(i,j)}p_r^{i_m,j} = 1$. The set of all $p_r^{i_m,j}$ is denoted as $\boldsymbol{p}$.

\item \textbf{Path flow} (denoted as $q^{i_m,j_n}_r$): Number of people who entry at station $i$ within time interval $m$ and exit at station $j$ within time interval $n$ using path $r$.

\item \textbf{Delay rate} (denoted as $\mu^{i_m,j_n}_r$): The proportion of people who exit at time interval $n$ compared to the total number of people who entry at $i_m$ with destination $j$ using path $r$ (i.e. $\mu^{i_m,j_n}_r = q^{i_m,j_n}_r/\sum_n q^{i_m,j_n}_r$). It contains the information of how many people exiting the system at different time intervals, which only depends on the train schedule and left behind because we fixed the path $r$. Since the schedule is known, the delay rate can be seen as an indicator of left behind. For example, If people were left behind many times in path $r$, we might have $\mu^{i_m,j_{2}}_r > \mu^{i_m,j_1}_r$ because more proportion of people tend to leave the system later (at $j_2$), rather than earlier (at $j_1$). The set of all $p_r^{i_m,j}$ is denoted as $\boldsymbol{\mu}$.
\end{itemize}

By the definition of these variables, we have the following relationships. 
\begin{itemize}
\item OD entry-exit flow equals the path flow of corresponding OD pairs sum over all paths.
\setlength{\belowdisplayskip}{7pt}
\setlength{\abovedisplayskip}{7pt}
\begin{flalign}
q^{i_m,j_n} = \sum_{r \in \mathscr{R}(i,j)} q_r^{i_m,j_n},  \; \;\;\; \forall i_m, j_n
\label{eq_odentryexit}
\end{flalign}
This relationship is trivial and directly hold by definition. Since we can observe the true OD entry-exit flow from AFC data, this equation connects the observed information with estimated information.

\item Path flow equals the OD entry flow times the path share times the delay rate.
\setlength{\belowdisplayskip}{7pt}
\setlength{\abovedisplayskip}{7pt}
\begin{flalign}
q^{i_m,j_n}_r = q^{i_m,j} \cdot p_r^{i_m,j} \cdot \mu_r^{i_m,j_n}, \; \;\;\;\forall i_m, j_n,r \in \mathscr{R}(i,j)
\label{eq_assign}
\end{flalign}
This equation is the major procedure for the network loading, which assigns the OD demand (OD entry flow) to the path flow. 

\end{itemize}

To better illustrate the network representation, we show a simple example below. Consider two stations, $i$ and $j$, of a metro system, where $i$ is the origin and $j$ is the destination (see Figure \ref{fig_phy}). Assume there exists two different paths connecting this OD pair, i.e. $\mathscr{R}(i,j) = \{1,2\}$. The red arrows represent path 1, and blue arrows represent path 2. Let us consider the studying time period from 7:00 to 7:30 and set the time interval $\tau=15 \text{ min}$. Then the physical network can be extended to the TS network shown in Figure \ref{fig_TS}. For example, $i_1$ here represents the station $i$ at time 7:00-7:15. Assume the only OD entry flow is $q^{i_1,j} = 10$, and the path shares are $p_1^{i_1,j} = 0.3$ and $p_2^{i_1,j} = 0.7$. Then we know there are totally 10 people entry station $i$ during 7:00-7:15. 3 of them use path 1 and 7 of them use path 2. They all head to destination $j$ but currently we do not know when they will arrive the destination. Actually, given current information, if we run a transit assignment model, it will tell us when the passengers exit the system. For illustration purpose, suppose we have the additional information of delay rate. For the 3 people who use path 1, assume 2 out of 3 people tap out at station $j$ during 7:00-7:15 (i.e. $\mu_1^{i_1,j_1} = 2/3$) and 1 out of 3 people tap out at station $j$ during 7:15-7:30 (i.e. $\mu_1^{i_1,j_2} = 1/3$). Then we have: $q_1^{i_1,j_1}=q^{i_1,j} \cdot p_1^{i_1,j} \cdot \mu_1^{i_1,j_1} =2$ and  $q_1^{i_1,j_2} = q^{i_1,j} \cdot p_1^{i_1,j} \cdot \mu_1^{i_1,j_2} =1$. These equations are exactly the examples of Eq. (\ref{eq_assign}), which assign the OD entry flow ($q^{i_1,j}$) to the path flow ($q_1^{i_1,j_1}$ and $q_1^{i_1,j_2}$). Similarly, for the 7 people who use path 2, assume 4 out of 7 people tap out at station $j$ during 7:00-7:15 (i.e. $\mu_2^{i_1,j_1} = 4/7$) and 3 out of 7 people tap out at station $j$ during 7:15-7:30 (i.e. $\mu_2^{i_1,j_2} = 3/7$). Then we have $q_2^{i_1,j_1} = 4$ and $q_2^{i_1,j_2} = 3$. 

From the relationship between OD entry-exit flow and path flow (Eq. (\ref{eq_odentryexit})), we have $q^{i_1,j_1} = q_1^{i_1,j_1} + q_2^{i_1,j_1} = 6$, and $q^{i_1,j_2} = q_1^{i_1,j_2} + q_2^{i_1,j_2} = 4$. Also, if we sum the OD entry-exit flow over exit time intervals, we will get OD entry flow, i.e. $q^{i_1,j} = 10 = q^{i_1,j_1} + q^{i_1,j_2}$.

\begin{figure}[H]
\centering
\subfloat[Physical network]{\includegraphics[width=0.38\textwidth]{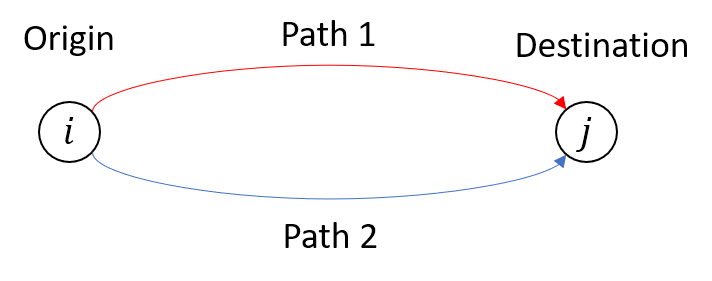}\label{fig_phy}}
\hfil
\subfloat[Time-space Hypernetwork]{\includegraphics[width=0.44\textwidth]{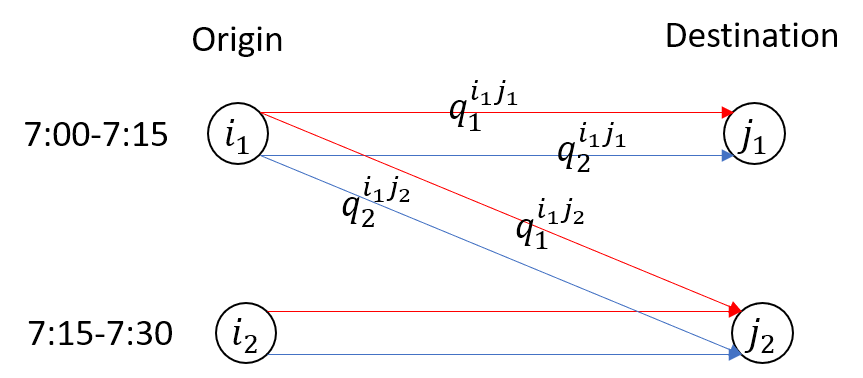}\label{fig_TS}}

\caption{Network Representation Example}
\end{figure}

\subsection{Problem Definition}
\vspace{3pt}
\subsubsection{Model assumptions}\label{assump}
Two major assumptions are made within the model and are presented below. First, we assume path shares can be formulated as a C-logit model \cite{cascetta1996modified}, which is an extension of multinomial logit (MNL) model. The formulation is shown below.
\begin{flalign}
p^{i_m,j}_r = \frac{\exp{(\beta_X \cdot X_{r,m}+\beta_{CF} \cdot CF_r)}}{\sum_{r' \in \mathscr{R}(i,j)} \exp{(\beta_X \cdot X_{r',m} + +\beta_{CF} \cdot CF_{r'})}} \vcentcolon= \frac{\exp{(\beta Y_{r,m})}}{\sum_{r' \in \mathscr{R}(i,j)} \exp{(\beta Y_{r',m})}} , 
\label{eq_MNL}
\end{flalign}
where $X_{r,m}$ are the attributes for path $r$ in time interval $m$, which include in-vehicle time, number of transfers, transfer walking time, etc.. $CF_r$ is the commonality factor of path $r$ which measures the degree of similarity of path $r$ with the other paths of the same OD. $\beta_X$ and $\beta_{CF}$ are the corresponding coefficients to be estimated. For simplicity, we define the $\beta$ and $Y_{r,m}$ as the combination of the two items in the utility function.  The $CF_r$ can be expressed as following.
\begin{flalign}
CF_r = \ln \sum_{r' \in \mathscr{R}(i,j)}(\frac{L_{r,r'}}{L_rL_{r'}})^\gamma, 
\label{eq_CF}
\end{flalign}
where $L_{r,r'}$ is the number of common stations of path $r$ and $r'$. $L_r$ and $L_{r'}$ are the number of stations for path $r$ and $r'$, respectively. $\gamma$ is a positive parameter which is assigned to 5 in this study. 

C-logit model can address the route overlapping problem with the correction term $CF$, which is widely used in modeling route path choices \citep{prato2009route}. Also, it remains the formulation of MNL, which is practical and easy to compute. 

Another assumption is related to the network loading criteria (NLC). We assume the following rules.
\begin{itemize}
\item When loading a train, passengers waiting at the platform are loaded based on a First-In-First-Board (FIFB) principle.
\item Every train has a strict physical capacity. When on-board passengers reach the capacity, the remained passengers will be left behind and wait in the platform for next available train. 
\end{itemize}

To make the network loading process (Eq. (\ref{eq_assign})) satisfies the NLC, $\mu_r^{i_m,j_n}$ has to satisfy some comprehensive constraints. This is because $\mu_r^{i_m,j_n}$ represents when and how many people exit the system, which is the reflection of NLC. Thus, the constraints of NLC should be formulated as the relationship among all $\mu_r^{i_m,j_n}$. However, it is difficult to formulate the constraints of $\mu_r^{i_m,j_n}$ analytically based on the aggregated network representation\footnote{It can only be done with the fine-grained TS network formulation, where $\tau$ is set as the minimum headway.}. Therefore, we temporally denote the constraints for $\mu_r^{i_m,j_n}$ as
\begin{flalign}
\mu_r^{i_m,j_n} \text{ satisfies the NLC}, \;\;\;\forall i_m, j_n,r \in \mathscr{R}(i,j)
\label{eq_mu_cons}
\end{flalign}
which is a non-analytical constraints and will be addressed in following sections.

Other minor assumptions are also made in this study.
\begin{itemize}
\item All transit services arrive on time. Time table is sufficiently reliable and can be considered as deterministic \citep{hamdouch2008schedule}. This assumption can be relaxed when the automated vehicle location (AVL) data is available, which can provide the ground-truth train arrival and departure time. 
\item The distribution of access walking time, egress walking time and transfer walking time are known. We can calculate when passengers arrive at the platform after tapping into the gates and when passengers get off the train before tapping out of the gates. 
\item The platform has infinite capacity to accommodate waiting passengers

\end{itemize}

\subsubsection{Formulation}
The purpose of this research is to estimate path choice using AFC data. Since we assume the path choice can be formulated as the C-logit model, the $\beta$ in C-logit model will be the decision variable. As we mentioned before, the OD entry-exit flow ($q^{i_m,j_n}$) is the output of transit assignment model, for which the ground truth value can also be observed from AFC data. So minimizing the difference between estimated and observed OD entry-exit flow can be the objective function. In case the prior information of path choice is available, we can add the difference between estimated $\beta$ and prior $\beta$ into the objective function as well. Therefore, we formulate the original problem as Eq. (\ref{eq:Orignal_prob}).

Eq. (\ref{eq:obj}) is the objective function, where $\tilde{q}^{i_{m} ,j_{n} }$ is the observed OD entry-exit flow; $\tilde{\beta}$ is the prior knowledge about the value of $\beta$, which can be the survey results from previous years. $w_1$ and $w_2$ are the corresponding weights. Note that in our case study section we assume no prior knowledge is known thus set $w_2 = 0$. Constraints \ref{eq:con1} and \ref{eq:con2} are the relationships described in Section \ref{network_repre}. Constraints \ref{eq:con3} and \ref{eq:con4} are the assumptions we made in Section \ref{assump}. Constraints \ref{eq:con5}, \ref{eq:con6} and \ref{eq:con7} are given by definition.

There are several constraints which make this problem hard to solve. First, constraints \ref{eq:con2} and \ref{eq:con3} are both nonlinear equality constraints because $\mu_r^{i_m,j_n}$ and $p_r^{i_m,j_n}$ are unknown. Second, constraint \ref{eq:con4} is non-analytical because we cannot formulate the NLC in terms of  $\mu_r^{i_m,j_n}$ analytically. So the original problem is intractable. The methods to deal with these constraints and approximately solve this original problem will be shown in following.   

\begin{mini!}|s|[2]                   % mini! = minimize 
    {\beta,\boldsymbol{\mu}}                               % optimization variable
    {w_1\sum_{i_m,j_n}(q^{i_{m}  ,j_{n}}-\tilde{q}^{i_{m} ,j_{n} })^2 + w_2||\beta - \tilde{\beta}||^2 \label{eq:obj}}   % objective function and label
    {\label{eq:Orignal_prob}}             % label for optimizatio problem
    {}                                % optimization result
    \addConstraint{{q}^{i_{m} ,j_{n}}}{=\sum_{r}{q}_r^{i_{m} ,j_{n}} }{\quad \quad \forall i_{m} ,j_{n} \label{eq:con1}}    % constraint 1
    \addConstraint{q^{i_{m} ,j_n}_r}{=  q^{i_{m} ,j} \cdot p^{i_m ,j}_r\cdot \mu^{i_{m} ,j_n}_r}{\quad \quad \forall i_{m} ,j, r\in \mathscr{R}(i,j) \label{eq:con2}}  % constraint 2
    \addConstraint{p_r^{i_m,j}}{=\frac{\exp{(\beta Y_{r,m})}}{\sum_{r' \in \mathscr{R}(i,j)} \exp{(\beta Y_{r',m})}}}{\quad \quad \forall i_{m} ,j, r\in \mathscr{R}(i,j) \label{eq:con3}}
    \addConstraint{\mu_r^{i_m,j_n} \text{ satisfies the NLC}}{}{\quad \quad \forall i_{m} ,j, r\in \mathscr{R}(i,j) \label{eq:con4}}
    \addConstraint{ \sum_{r\in \mathscr{R}(i,j)}  p^{i_{m} ,j}_r =  1}{}{\quad \quad \forall i_{m} ,j \label{eq:con5}}
    \addConstraint{ 0 \leq p^{i_{m} ,j}_r  \leq 1}{}{\quad \quad \forall i_{m} ,j, r\in \mathscr{R}(i,j) \label{eq:con6}}
    \addConstraint{q^{i_{m} ,j_n}_r \geq 0}{}{\quad \quad \forall i_{m} ,j, r\in \mathscr{R}(i,j) \label{eq:con7}}
\end{mini!}

\subsection{Problem decomposition}
Since we cannot formulate the constraints of $\mu^{i_{m} ,j_n}_r$ analytically, a natural method is to derive it from the results of a network loading process. Therefore, we decompose the original problem into two sub-problems as following.

\begin{itemize}
\item Sub-problem 1:
\begin{mini}|s|[2]                 
    {\beta}                            
    {w_1\sum_{i_m,j_n}(q^{i_{m}  ,j_{n}}-\tilde{q}^{i_{m} ,j_{n} })^2 + w_2||\beta - \tilde{\beta}||^2 \label{eq:eq1}}{}{}                   
    \addConstraint{\text{Eq. (\ref{eq:con1}) - (\ref{eq:con3})}}
    \addConstraint{\text{Eq. (\ref{eq:con5}) - (\ref{eq:con7})}}
\end{mini}
\item Sub-problem 2:
\setlength{\belowdisplayskip}{7pt}
\setlength{\abovedisplayskip}{7pt}
\begin{flalign}
\boldsymbol{\mu} = \text{Network Loading }(\beta, \boldsymbol{q_e}, \theta) 
\end{flalign}
\end{itemize}

Sub-problem 2 is the network loading model, which takes the route choice parameter $\beta$, OD entry demand $\boldsymbol{q_e}$ and model parameters $\theta$ as input, and outputs the delay rate $\boldsymbol{\mu}$. In this study, we use an event-based simulation model proposed by \cite{Ma2019Network} to perform the network loading. The model parameters $\theta$ include time table (or AVL data), transit network typology, access/egress/transfer time and train capacity, which are assumed to be known. Two events are considered in the network loading model, one is train arrival event, in which we offload passengers who will transfer or exit in the station. Another is train departure event, in which we update the passengers in the platform and load  passengers into the train. All events are processed sequentially based on their occurrence time so as to simulate the whole network loading process. More information about the simulation model can be found in \cite{Ma2019Network}. This simulation framework shares the same NLC and model assumptions as we described before. Therefore, the estimated $\boldsymbol{\mu}$ from the model will naturally satisfy the NLC constraints.

Sub-problem 1 is the variation of the original problem (Eq. (\ref{eq:Orignal_prob})) where the non-analytical constraint \ref{eq:con4} is removed. Besides, $\boldsymbol{\mu}$ is treated as known constants in sub-problem 1. Therefore, after problem decomposition, we can iteratively solve these two sub-problems to approximate the solution of the original problem. Looking at the properties of sub-problem 1, constraint \ref{eq:con2} is now linear since $\mu^{i_{m} ,j_n}_r$ is fixed. Also, we do not have the non-analytical constraint. But it is still intractable because of the highly non-linear constraint \ref{eq:con3}, which we called \emph{MNL constraints} hereafter. In the following sections, we will show how we approximately linearize the sub-problem 1 so as to transfer it to a simple quadratic programming problem.

\subsection{Approximate linearization for sub-problem 1}
Addressing the MNL constraints is difficult in the literature. \cite{davis2013assortment} and \cite{atasoy2015concept} showed when the MNL structure is in the objective function, utilities are constants but choice sets are unknown, the integer programming can be reformulated as a linear programming. However, for our problem the MNL structure is in the constraint part and $\beta$ in the utility function are unknown. Based on the authors' knowledge, there is no equivalent transformation from this MNL constraint to a tractable form. In this study, we propose two procedures to approximately linearize the sub-problem 1 with MNL constraints. The word "approximate" means the problem after transformation is not equivalent to the original problem, but solving the new problem can obtain the results close to the original one. 
\subsubsection{Construct approximate linear constraints (ALC)}
The MNL constraint shows the relationship between $\beta$ and $p^{i_m ,j}_r$. Since directly dealing with the non-linear constraints is difficult, we first replace the decision variables $\beta$ with $p^{i_m ,j}_r$ and remove constraint \ref{eq:con3}. Then the sub-problem 1 will become a simple quadratic programming given all constraints become linear. However, as the degree of freedom for $p^{i_m ,j}_r$ is much larger than $\beta$, directly replacing decision variables will cause severe problems of over-fitting. So, we need more constraints on $p^{i_m ,j}_r$ to narrow down the feasible space.

In this study, we propose a Monte-Carlo sampling method to construct a series of linear constraints for $p^{i_m ,j}_r$. The basic idea is that, for some OD pairs with same path sets, the corresponding path choice fractions may be same under MNL constraints. So we can construct linear constraints with the form $p^{i_m ,j}_r = p^{{i'_{m'}} ,j'}_{r^{'}}$ for some $i_m,j,r$ and $i'_{m'},j',r'$. 

A simple example is shown below to illustrate this property. Consider the OD pairs 1-5 and 2-5 in Figure \ref{fig_ALC}. There are two paths for both OD pairs. Path 1 transfers at station 4 and path 2 transfers at station 3. The path choice fractions are denoted as $p^{1,5}_1, p^{1,5}_2, p^{2,5}_1, p^{2,5}_2$, respectively. Note that we ignore the time index here for simplification. Assume there are four path attributes affecting people's choice: in-vehicle time, number of transfers and transfer walking time and commonality factor. Then we have the following proposition.

\begin{prop}\label{prop_mnl_equal}
Under MNL constraints, there are $p^{1,5}_1 = p^{2,5}_1$ and $p^{1,5}_2 = p^{2,5}_2$.
\end{prop}

\begin{proof}
Denote the utility for path $r$ of OD pair $i$ and $j$ as $V_r^{i,j}$. Since the path 1 of OD 1-5 and the path 1 of OD 2-5 share the same transfer patterns, the number of transfers and transfer walking time for them are same. The only differences are the in-vehicle time and commonality factor. It tends out the commonality factors for these two paths are also same. The proof is easy to obtain by following Eq. (\ref{eq_CF}) and is ignored here. 
Therefore, the utility difference for path 1 of two different OD pairs only contains in-vehicle time term. Let the total in-vehicle time for path $r$ of OD pair $(i,j)$ be $tt_r^{i,j}$. Denote the in-vehicle time for link $(i,j)$ as $tt^{i,j}$. and the coefficients of in-vehicle time as $\beta_{tt}$. We have
\begin{align}
V_1^{1,5} -V_1^{2,5} = \beta_{tt} \cdot (tt_1^{1,5} - tt_1^{2,5}) = \beta_{tt} \cdot tt^{1,2}, \label{od1}
\end{align}
Similarly, for path 2 of  OD 1-5 and OD 2-5, we have
\begin{align}
V_2^{1,5} -V_2^{2,5} = \beta_{tt} \cdot (tt_2^{1,5} - tt_2^{2,5}) = \beta_{tt} \cdot tt^{1,2}, \label{od2}
\end{align}

According to the MNL constraint, we have
\begin{flalign}
p^{1,5}_1 &= \frac{1}{1+\exp(V_2^{1,5} - V_1^{1,5})} = \frac{1}{1+\exp((V_2^{2,5}+\beta_{tt} \cdot tt^{1,2}) - (V_1^{2,5}+\beta_{tt} \cdot tt^{1,2}))}. \\
&=\frac{1}{1+\exp(V_2^{2,5} - V_1^{2,5})} = p^{2,5}_1 \nonumber
\end{flalign}
Similarly, for path 2, we will have $p^{1,5}_2 = p^{2,5}_2$. 

This example network is extracted and simplified from the real-world Hong Kong MTR network. Therefore, this property does hold in the real-world for many OD pairs.

\begin{figure}[H]
\centering
\includegraphics[width=3.5 in]{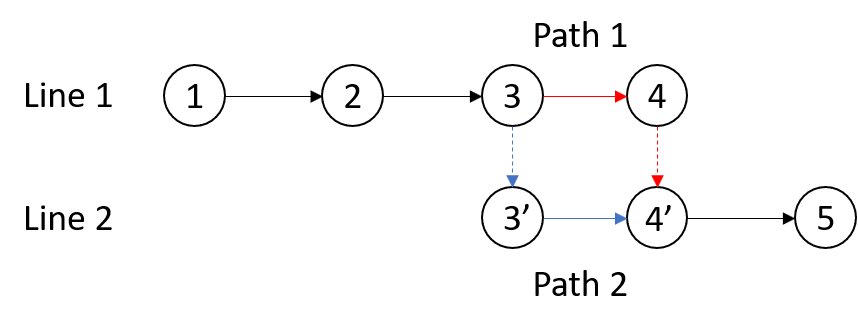}
\caption{Network Example for ALC}
\label{fig_ALC}
\end{figure}

\end{proof}

Besides equality constrains, there are also some inequality constraints hold under MNL constraint. For example, since all cost coefficients should be negative according our prior knowledge, if there is a path have smaller costs than other paths in the same OD pair, it should always be preferred regardless the magnitude of $\beta$. So we can construct linear constraints with the form $p^{i_m ,j}_r \geq  p^{{i_{m}} ,j}_{r^{'}}$ to capture this information. To automatically extract all these linear constraints in the system, we propose a Monte-Carlo sampling method. We first define a reasonable range for all $\beta$ (i.e. $\beta \in [L_\beta,U_\beta]$) based on the prior knowledge (e.g. survey results from previous years), where $L_\beta$ ($U_\beta$) is the lower (upper) bound of $\beta$. It is worth noting that the selection of $L_\beta$ and $U_\beta$ has very limited impact on the construction of ALC. The proof of Proposition \ref{prop_mnl_equal} tells us the equality constraints have nothing to do with the value of $\beta$. $L_\beta$ and $U_\beta$ only affects the construction of inequality constraints, and from our numerical tests the impact is very small. Generally, we only need to set the cost coefficients (e.g. in-vehicle time) to be negative.

The detailed ALC construction steps are shown in Algorithm \ref{alg_monte}. Denote the maximum number of sampling as $S$. The choose of $S$ is a trade off between computational efficiency and constraints accuracy. Larger $S$ can help avoid the coincidence of erroneously constructing the constraints.

\begin{algorithm} 
\caption{Monte-Carlo Based ALC Construction} \label{alg_monte}
\begin{algorithmic}[1]
\State Initialize $s=0$
\While{$s < S$}
\State $s=s+1$
\State Sample $\beta$ from the uniform distribution $U(a,b)$, denote them as $\beta^{(s)}$.
\State Calculate the path choice fraction for all paths based on $\beta^{(s)}$, denote them as ${p^{i_m ,j}_r}^{(s)}$
\EndWhile
\For{$\text{all } i_m,j,r \text{ in path sets}$} 
\For{$\text{all } {i_m}',j',r' \text{ in path sets}$}  
\If{${p^{i_m ,j}_r}^{(s)} = {p^{{i_{m}}' ,j'}_{r^{'}}}^{(s)}$ for all $s=1,...,S$}
\State Save ${p^{i_m ,j}_r} = {p^{{i_{m}}' ,j'}_{r^{'}}}$ as a linear constraint.
\EndIf 
\EndFor
\EndFor
\For{$\text{all } i_m,j \text{ in OD pairs sets}$} 
\For{$\text{all } r \in \mathscr{R}(i,j)$}  
\For{$\text{all } r' \in \mathscr{R}(i,j)$}
\If{${p^{i_m ,j}_r}^{(s)} \geq {p^{{i_{m}} ,j}_{r^{'}}}^{(s)}$ for all $s=1,...,S$}
\State Save ${p^{i_m ,j}_r} \geq {p^{{i_{m}} ,j}_{r^{'}}}$ as a linear constraint.
\EndIf
\EndFor
\EndFor
\EndFor
\State \Return All saved linear constraints
\end{algorithmic}
\end{algorithm}

We denote all the constructed linear constraints for ${p^{i_m ,j}_r} $ as 
\begin{flalign}
{p^{i_m ,j}_r} \text{ satisfies the ALC of MNL}, \;\;\;\forall i_m, j_n,r\in\mathscr{R}(i,j)
\label{eq_alc}
\end{flalign}
Thus, the sub-problem 1 can be reformulated as 

\begin{mini}|s|[2]                 
    {\boldsymbol{p}}
    {w_1\sum_{i_m,j_n}(q^{i_{m}  ,j_{n}}-\tilde{q}^{i_{m} ,j_{n} })^2 + w_2\sum_{i_m,j}\sum_{r \in \mathscr{R}(i,j)}(p^{i_m ,j}_r - \tilde{p}^{i_m ,j}_r)^2}
    {\label{eq:sub1a}}
    {}
    \addConstraint{\text{Eq. \ref{eq:con1} - \ref{eq:con2}}}
    \addConstraint{{p^{i_m ,j}_r} \text{ satisfies the ALC of MNL}}{}{\;\;\;\forall i_m, j_n,r \in \mathscr{R}(i,j)}
    \addConstraint{\text{Eq. \ref{eq:con5} - \ref{eq:con7}}}
\end{mini}
Note that we replace the decision variables from $\beta$ to $\boldsymbol{p}$. $\tilde{p}^{i_m ,j}_r$ is the prior knowledge about path share derived from $\tilde{\beta}$. Clearly, Eq. (\ref{eq:sub1a}) is a quadratic programming since all constraints are linear. It now can be solved very fast. However, despite Eq. (\ref{eq_alc}) add some constraints to $p^{i_m ,j}_r$, it is not equivalent to (less strong than) the original MNL constraints. This is why we called it "approximate" linear constraints (ALC). After adding the ALC, $\boldsymbol{p}$ may still have high degree of freedom. Based on the numerical test in the case study, after adding the ALC, the total degree of freedom can decrease around 40\%, which demonstrates a narrower feasible space. But we still need to go one step further to make all estimated path shares satisfy the MNL constraints. 

\subsubsection{MNL correction}
The estimated $\boldsymbol{p}$ from Eq. \ref{eq:sub1a} have two problems. The first is possible over-fitting due to high degree of freedom, which we have discussed before. The second is inestimable path shares due to little observed OD entry-exit flow. For example, if there is no passenger observed for OD pair $(i,j)$ in time interval $m$. $p^{i_m ,j}_r$ can take any values and does not affect the objective function, which makes it unable to be estimated. Both of these problems can be attributed to one reason: the estimated $p^{i_m ,j}_r$ violates the original MNL constraints (it only satisfies the ALC of MNL). 

To address this problem, we can use the estimated $p^{i_m ,j}_r$ from Eq. (\ref{eq:sub1a}) (called \emph{rough} path shares hereafter) to obtain a set of $\beta$, and then use the $\beta$ to generate new path shares. This procedure is called MNL correction. Path shares after MNL correction will naturally satisfy MNL constraints by definition. However, not all rough path shares are equally reliable. We have the following proposition for the reliability discussion.

\begin{prop}\label{prop_reliability}
The reliability of estimated $p^{i_m ,j}_r$ from Eq. (\ref{eq:sub1a}) is proportional to the corresponding OD entry flow $q^{i_m ,j}$.
\end{prop}
\begin{proof}
Intuitively, more observed passengers can provide more information for the path shares estimation. To validate it mathematically, we observe 
\begin{flalign}
q^{i_m,j_n} = q^{i_m,j} \sum_{r \in \mathscr{R}(i,j)}  p_r^{i_m,j} \cdot \mu_r^{i_m,j_n}, \; \;\;\;\forall i_m, j_n
\label{eq_comb}
\end{flalign}
by combining Eq. (\ref{eq_odentryexit}) and (\ref{eq_assign}). The reliability of $p^{i_m ,j}_r$ can be measured by the estimation variance. Thus,
\begin{flalign}
\text{Var}\left[\sum_{r \in \mathscr{R}(i,j)}  p_r^{i_m,j} \cdot \mu_r^{i_m,j_n}\right] = \frac{\text{Var}[q^{i_m,j_n}]}{(q^{i_m,j})^2}, \; \;\;\;\forall i_m, j_n
\label{eq_var_equal}
\end{flalign}

Since $q^{i_m,j_n}$ is the quadratic term in the objective function of Eq. (\ref{eq:sub1a}), analogue to the simple least square regression, we can assume $\text{Var}[q^{i_m,j_n}] = \sigma^2$ for all $i_m, j_n$ (i.e. homoscedasticity). This leads to  
\begin{flalign}
\text{Var}\left[\sum_{r \in \mathscr{R}(i,j)}  p_r^{i_m,j}\right] \propto \frac{\sigma^2}{(q^{i_m,j})^2}, \; \;\;\;\forall i_m, j_n
\label{eq_var_prop}
\end{flalign}
Note that we ignore $\mu_r^{i_m,j_n}$ since it is constant in sub-problem 1. Recall the problem that if there is no passengers observed for a specific OD pair ($q^{i_m,j} = 0$), the corresponding path share $ p_r^{i_m,j}$ is inestimable in Eq. (\ref{eq:sub1a}). Based on our derivation in Eq. (\ref{eq_var_prop}), this scenario results in the variance of estimated path shares equal to $+\infty$, which validates our derivation.

Notably, we only obtain the estimation variance of $\sum_{r \in \mathscr{R}(i,j)}  p_r^{i_m,j}$, rather than $p_r^{i_m,j}$. One reason is that $\text{Var}[p_r^{i_m,j}]$ is hard to derive. Another is that the following estimation of $\beta$ based on $\boldsymbol{p}$ can utilize $\text{Var}\left[\sum_{r \in \mathscr{R}(i,j)}  p_r^{i_m,j}\right]$ as a whole, that is, use $\frac{1}{\sqrt{\text{Var}\left[\sum_{r \in \mathscr{R}(i,j)}  p_r^{i_m,j}\right]}} \propto q^{i_m,j}$ as the weights, which eliminates the needs for deriving $\text{Var}[p_r^{i_m,j}]$.

\end{proof}

We formulate the MNL correction problem as following, which can be seen as a weighted fractional logit model \citep{papke1996econometric}.
\begin{flalign}
\label{eq_est_beta}
\max\limits_{\beta}  \sum_{i_m,j} q^{i_m,j} \sum_{r \in \mathscr{R}(i,j)}  p^{i_m ,j}_r \cdot \log \frac{ \exp(\beta Y_{r,m})}{\sum_{r' \in \mathscr{R}(i,j)}\exp(\beta Y_{r',m})}
\end{flalign}
Note that in Eq. (\ref{eq_est_beta}), $p^{i_m ,j}_r$ are constants. The objective function has the form of softmax function. Thus it is a convex optimization problem without constraints (like logistic regression), which can be solved efficiently. $q^{i_m,j}$ is the weight for corresponding path shares ($ \sum_{r \in \mathscr{R}(i,j)}  p^{i_m ,j}_r $), which reflects their reliability as we discussed in Proposition \ref{prop_reliability}. After we get $\beta$, the aforementioned two problems in results of Eq. (\ref{eq:sub1a}) will naturally disappear, because we can generate a new $\boldsymbol{p}$ which satisfies the MNL constraints exactly.

\subsection{Discussion of solution procedures}
So far, we have formulated three sub-problems to approximate the solution for the original problem. These sub-problems can be summarized in Eq. (\ref{eq_sub1a})-(\ref{eq_sub2}). In sub-problem 1(a), given the $\mu^{i_{m} ,j_n}_r$, we estimate the rough path shares, which is a quadratic programming problem. In sub-problem 1(b), given the rough path shares, we estimate the corresponding $\beta$, which can be seen as a weighted fractional logit model. In sub-problem 2, given the $\beta$, we load passengers to the network and return the $\mu^{i_{m} ,j_n}_r$ which satisfies the NLC constraints. 

\vspace{3pt}
\begin{itemize}
\item Sub-problem 1(a):
\begin{mini}|s|[2]                 
    {\boldsymbol{p}}
    {w_1\sum_{i_m,j_n}(q^{i_{m}  ,j_{n}}-\tilde{q}^{i_{m} ,j_{n} })^2 + w_2\sum_{i_m,j}\sum_{r \in \mathscr{R}(i,j)}(p^{i_m ,j}_r - \tilde{p}^{i_m ,j}_r)^2}
    {\label{eq_sub1a}}
    {}
    \addConstraint{\text{Eq. (\ref{eq:con1}) - (\ref{eq:con2})}}
    \addConstraint{{p^{i_m ,j}_r} \text{ satisfies the ALC of MNL}}{}{\;\;\;\forall i_m, j,r \in \mathscr{R}(i,j)}
    \addConstraint{\text{Eq. (\ref{eq:con5}) - (\ref{eq:con7})}}
\end{mini}

\item Sub-problem 1(b):
\setlength{\belowdisplayskip}{7pt}
\setlength{\abovedisplayskip}{7pt}
\begin{flalign}
\label{eq_sub1b}
\max\limits_{\beta}  \sum_{i_m,j} q^{i_m,j} \sum_{r \in \mathscr{R}(i,j)}  p^{i_m ,j}_r \cdot \log \frac{ \exp(\beta Y_{r,m})}{\sum_{r' \in \mathscr{R}(i,j)}\exp(\beta Y_{r',m})}
\end{flalign}

\item Sub-problem 2:
\setlength{\belowdisplayskip}{7pt}
\setlength{\abovedisplayskip}{7pt}
\begin{flalign}
\label{eq_sub2}
\boldsymbol{\mu} = \text{Network Loading }(\beta, \boldsymbol{q_e}, \theta) 
\end{flalign}
\end{itemize}

We expect to solve these three sub-problems iteratively and approximate the solution for the original problem. This is equivalent to find a fixed point of the following problem.
\begin{flalign}
\label{eq_fixed_point}
\beta = \textsc{SP1b} \circ \textsc{SP1a} \circ \textsc{SP2}(\beta)
\end{flalign}
where \textsc{SP2} is the solution function of Sub-problem 2, i.e. $\boldsymbol{\mu} = \textsc{SP2}(\beta)$; \textsc{SP1a} is the solution function of Sub-problem 1(a), i.e. $\boldsymbol{p} = \textsc{SP1a}(\boldsymbol{\mu})$;  \textsc{SP1b} is the solution function of Sub-problem 1(b), i.e. $\beta = \textsc{SP1b}(\boldsymbol{p})$; "$ \circ $" is the sign of function composition, i.e., $f \circ g (x) =  f(g(x))$. We are curious about the existence and uniqueness of the solution in Eq. (\ref{eq_fixed_point}), and its relationship with the original problem in Eq. (\ref{eq:Orignal_prob}). Before discussing the main proposition, a lemma is introduced. 

\begin{lemma}\label{lemma_feasibility}
If a path share set $\boldsymbol{p}$ satisfies MNL constraints in terms of some choice parameters $\beta^*$, then $\beta^*$ is the solution of sub-problem 1(b) with respect to $\boldsymbol{p}$. Mathematically, if $p^{i_m ,j}_r = \frac{\exp{(\beta^* Y_{r,m})}}{\sum_{r' \in \mathscr{R}(i,j)} \exp{(\beta^* Y_{r',m})}}$ for all $p^{i_m ,j}_r \in \boldsymbol{p}$, then $\beta^* = \textsc{SP1b}(\boldsymbol{p})$.
\end{lemma}

\begin{proof}
Denote $\frac{\exp{(\beta Y_{r,m})}}{\sum_{r' \in \mathscr{R}(i,j)} \exp{(\beta Y_{r',m})}}$ as $h^{i_m ,j}_r$. Then Eq. \ref{eq_sub1b} can be rewritten as
\begin{flalign}
\label{eq_lemma1_proof}
\max\limits_{\beta}  \sum_{i_m,j} q^{i_m,j} \sum_{r \in \mathscr{R}(i,j)}  p^{i_m ,j}_r \cdot \log h^{i_m ,j}_r,
\end{flalign}
which has the form of entropy function. The maximum can be reached when $p^{i_m ,j}_r = h^{i_m ,j}_r, \; \forall i_m, j,r \in \mathscr{R}(i,j) $. Since we already have $p^{i_m ,j}_r$ satisfies MNL constraints in terms of $\beta^*$, feeding $\beta^*$ into Eq. (\ref{eq_lemma1_proof}) gives the desired condition for $p^{i_m ,j}_r$ and $ h^{i_m ,j}_r$. Thus $\beta^*$ is the optimal solution of sub-problem 1(b).
\end{proof}

% \begin{lemma}\label{lemma_Banach}
% [Banach fixed-point theorem] Let $(\mathcal{X},d)$ be a non-empty complete metric space with a mapping $T: \mathcal{X} \to \mathcal{X}$. $d$ is a distance measure (norm) function. If there exists a $\lambda \in [0,1)$ such that $d(T(x),T(y))\leq \lambda d(x,y)$ for all $x,y$ in $\mathcal{X}$. Then $T$ admits a unique fixed-point $x^*$ in $\mathcal{X}$ (i.e. $T(x^*) = x^*$). Furthermore, $x^*$ can be found as follows: start with an arbitrary element $x_0$ in $\mathcal{X}$ and define a sequence $\{x_n\}$ by $x_n = T(x_{n-1})$ for $n \geq 1$. Then $\lim_{n\to\infty} x_n = x^*$.
% \end{lemma}

% Banach fixed-point theorem gives the sufficient condition of the uniqueness of a fixed point. The proof can be found in the book of \cite{luan2015theorem}. 

The discussion of existence of fixed point in Eq. (\ref{eq_fixed_point}) is followed. Here, we assume the objective function equal to 0 is reachable by some $\beta$ in the original problem (Eq. (\ref{eq:Orignal_prob})), which is an ideal situation where all passengers' behaviors are assumed to be perfectly described by the C-logit model. Extending the proof to more complicated scenarios need future works. 
\begin{prop}\label{prop_fixed_point}
If the objective function equal to 0 is reachable by some $\beta$ in the original problem (Eq. (\ref{eq:Orignal_prob})), then the optimal solution $\beta^*$ for the original problem is one of the fixed point for Eq. (\ref{eq_fixed_point}). 
\end{prop}

\begin{proof}
Denote $\boldsymbol{\mu}^* = \textsc{SP2}(\beta^*)$. Define $\boldsymbol{p}^*$ such that ${p^{i_m ,j}_r}^* = \frac{\exp{(\beta^* Y_{r,m})}}{\sum_{r' \in \mathscr{R}(i,j)} \exp{(\beta^* Y_{r',m})}}$ for all ${p^{i_m ,j}_r}^* \in \boldsymbol{p}^*$. We claim $\boldsymbol{p}^* = \textsc{SP1a}(\boldsymbol{\mu}^*)$. The proof is shown below.

Clearly, $\boldsymbol{\mu}^*$ satisfies the NLC and $\boldsymbol{p}^*$ satisfies the MNL constraints. So ($\beta^*, \boldsymbol{\mu}^*, \boldsymbol{p}^*$) is the optimal solution for the original problem. Comparing sub-problem 1(a) and the original problem, if we feed in $\boldsymbol{\mu}^*$ into sub-problem 1(a), the optimal objective function of sub-problem 1(a) should be less than or equal to that of original problem because $\boldsymbol{p}$ has larger feasible space in sub-problem 1(a). However, given the optimal objective function of original problem is 0, we know the optimal objective function of sub-problem 1(a) is 0 as well. And since the objective function for these two problems are same, we have $\boldsymbol{p}^*$ is the optimal solution for sub-problem 1(a) as well.

By definition, $\boldsymbol{p}^*$ satisfies the MNL constraints in terms of $\beta^*$. According to Lemma \ref{lemma_feasibility}, we have $\beta^* = \textsc{SP1b}(\boldsymbol{p^*})$. This leads to $\beta^*  = \textsc{SP1b} \circ \textsc{SP1a} \circ \textsc{SP2}(\beta^*)$. 
% \textbf{[Uniqueness]}: Uniqueness can be proved numerically. Start with an arbitrary $\beta_0$ and define a sequence $\{\beta_n\}$ by $\beta_n = \textsc{SP1b} \circ \textsc{SP1a} \circ \textsc{SP2}(\beta_{n-1})$ for $n \geq 1$. By the contrapositive of Lemma \ref{lemma_Banach}, if there does NOT exist a $\beta^*$ such that $\lim_{n\to\infty} \beta_n = \beta^*$, then the fixed point is not unique or not exists. We tested whether the limit of $\beta_n$ exists or not numerically with the synthetic data described the Section \ref{case_study}, and found there exists fluctuations for the estimated parameters (see Figure \ref{fig_beta_fluc}). Combining with the proof of existence, we conclude the fixed point for Eq. (\ref{eq_fixed_point}) is not unique.
% \begin{figure}[H]
% \centering
% \includegraphics[width=4.5 in]{Figures/beta_fluc_example.png}
% \caption{Fluctuation in estimated parameters}
% \label{fig_beta_fluc}
% \end{figure}
\end{proof}
The proof of uniqueness requires the application of Banach fixed-point theorem \citep{luan2015theorem}. However, since $\textsc{SP1b} \circ \textsc{SP1a} \circ \textsc{SP2}$ has no analytical expression, it is hard to prove the uniqueness religiously. A corollary of Banach fixed-point theorem provides a way to obtain the fixed point: Start with an arbitrary $\beta_0$ and define a sequence $\{\beta_n\}$ by $\beta_n = \textsc{SP1b} \circ \textsc{SP1a} \circ \textsc{SP2}(\beta_{n-1})$ for $n \geq 1$. If there exists a $\beta^*$ such that $\lim_{n\to\infty} \beta_n = \beta^*$. Then $\beta^*$ is the fixed point. Following this corollary, we can test the convergence of $\{\beta_n\}$ numerically based on the synthetic data described in Section \ref{case_study}. The results are shown in Figure \ref{fig_beta_convergence}. All $\beta$ has shown convergence trends despite of some slight fluctuation in the tail. The fluctuation may come from the randomness in the network loading model. 

\begin{figure}[H]
\centering
\subfloat[In-vehicle time]{\includegraphics[width=0.25\textwidth]{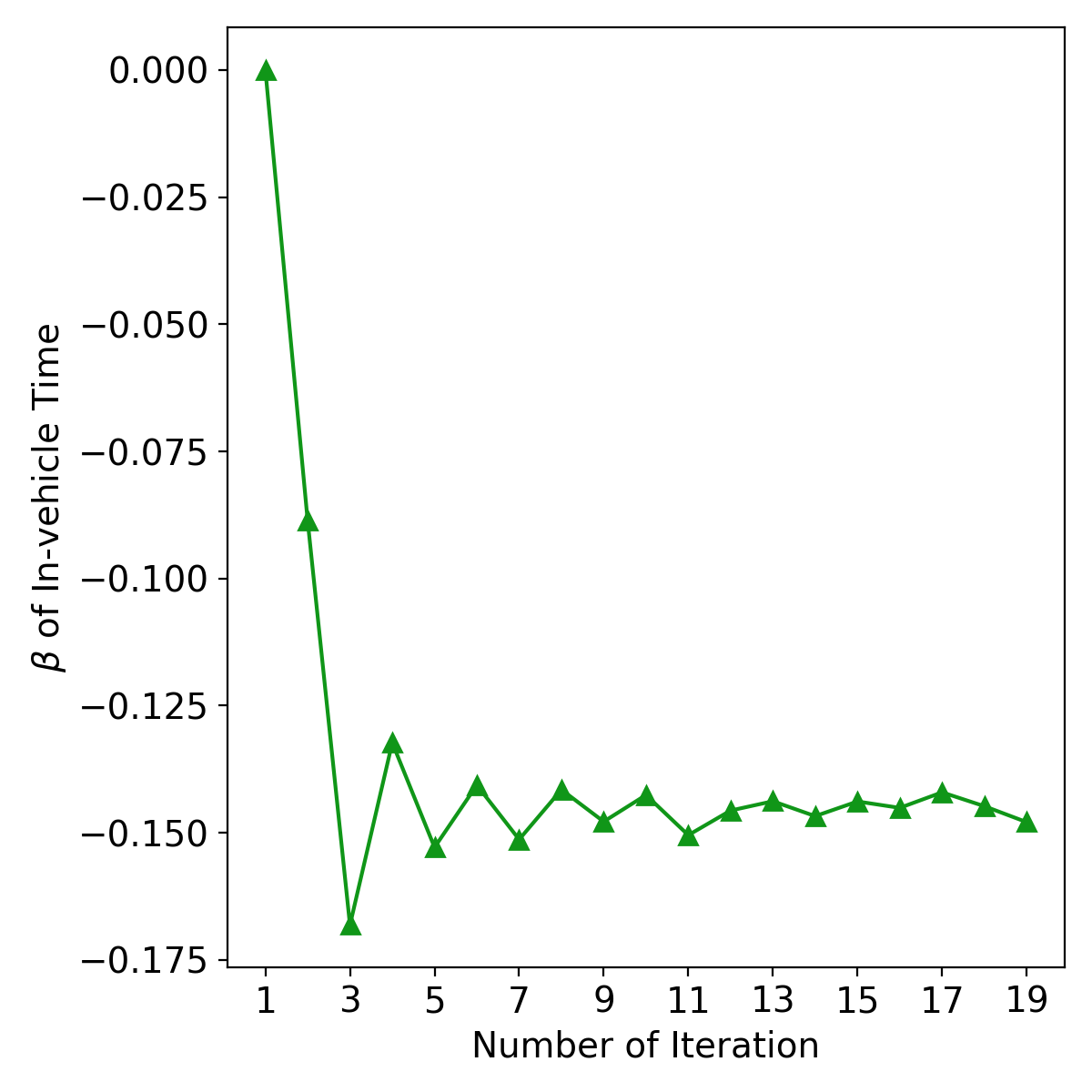}\label{fig_beta_1}}
\hfil
\subfloat[Num of transfer]{\includegraphics[width=0.25\textwidth]{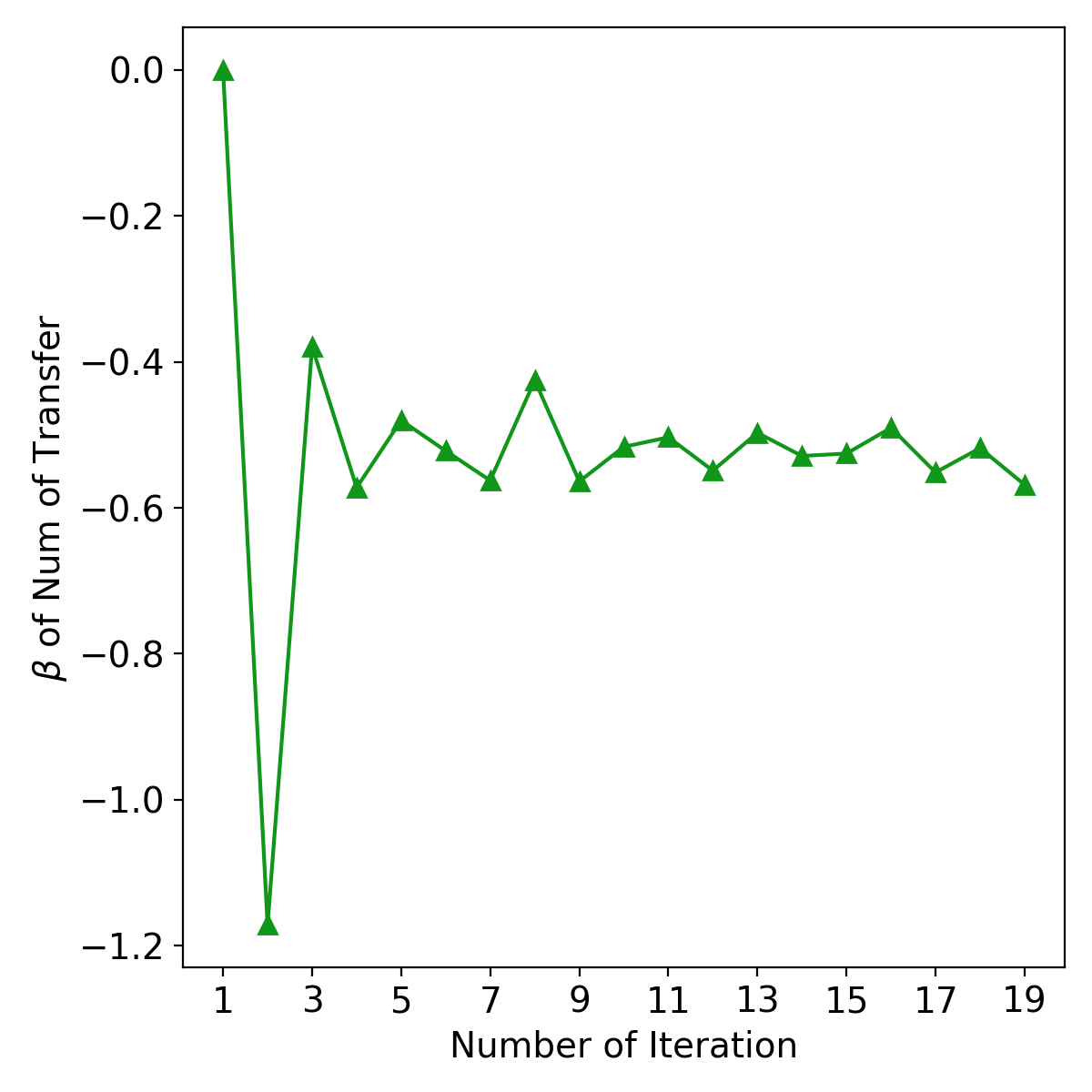}\label{fig_beta_2}}
\hfil
\subfloat[Relative walking time]{\includegraphics[width=0.25\textwidth]{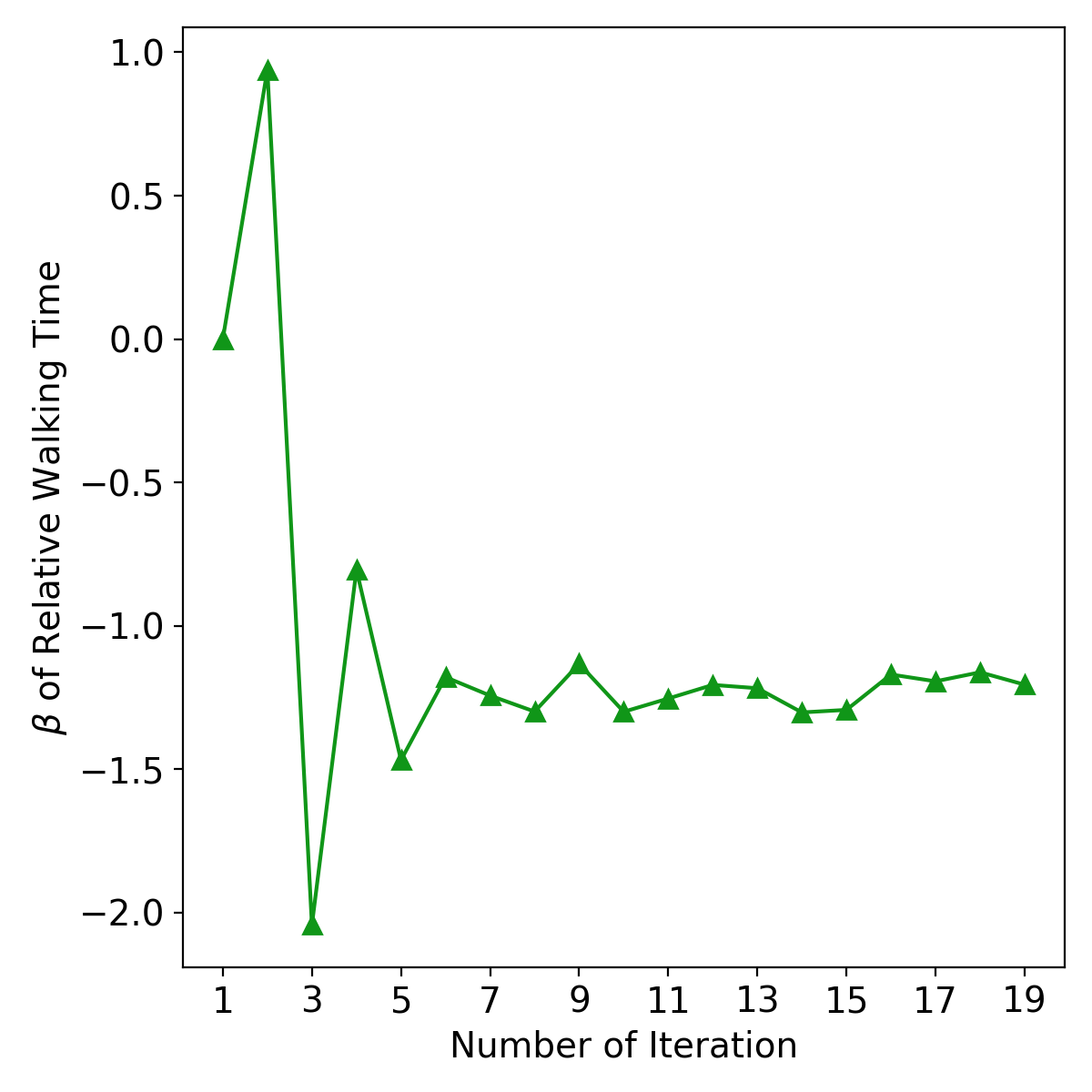}\label{fig_beta_3}}
\hfil
\subfloat[Commonality factor]{\includegraphics[width=0.25\textwidth]{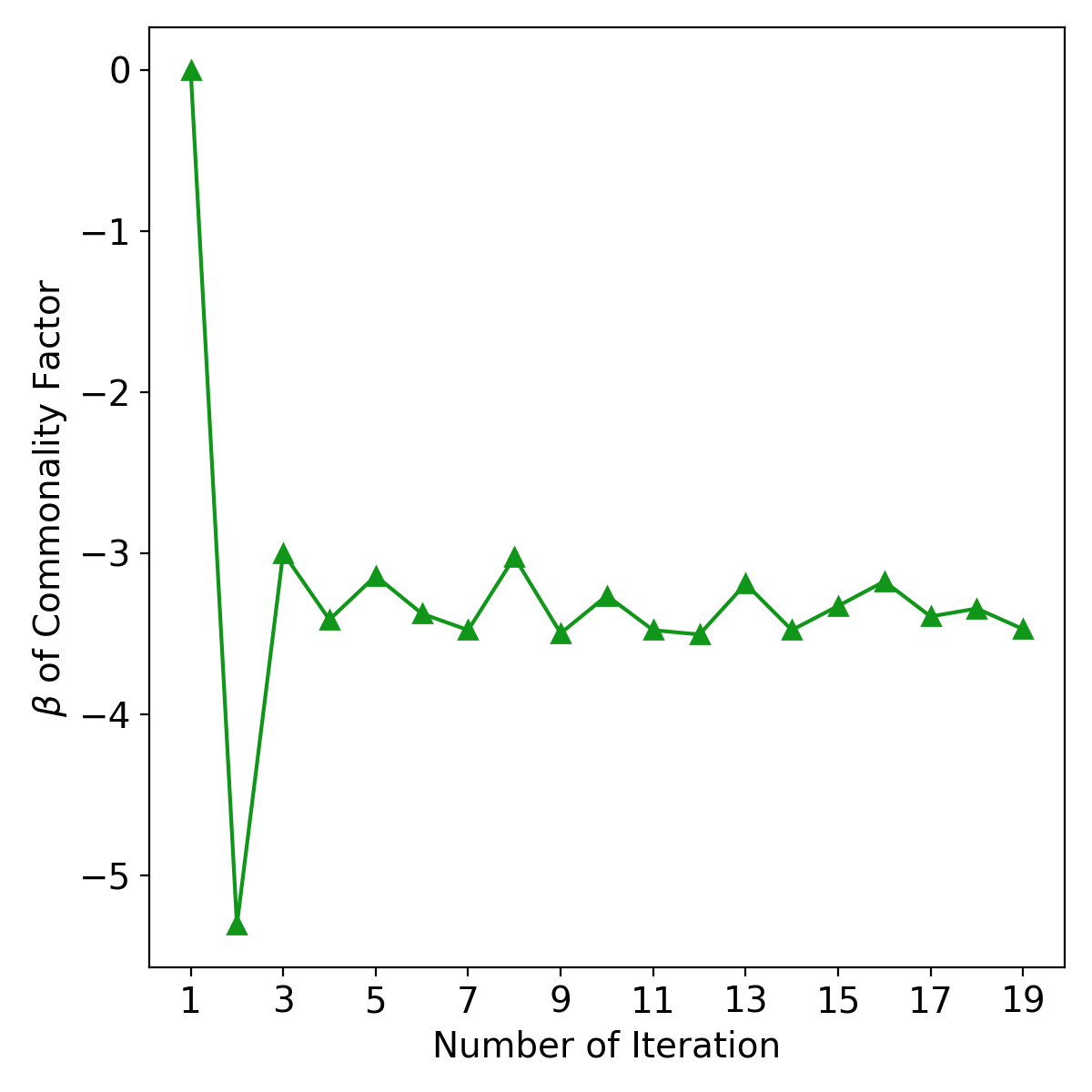}\label{fig_beta_4}}
\caption{Convergence Testing of Estimated $\beta$}
\label{fig_beta_convergence}
\end{figure}

% The proof of Proposition \ref{prop_fixed_point} results in two problems: (1) 
% Since  Therefore, a method of successful average (MSA) is introduced to the iteration. Denote the delay rate directly obtained from sub-problem 2 in $k$-th iteration as ${{}{\tilde{\mu}^{i_{m} ,j_n}_r}}^{(k)}$, and the delay rate input to sub-problem 1(a) in $k$-th iteration as ${{\mu}^{i_{m} ,j_n}_r}^{(k)}$. Let the average rate for $k$-th iteration be $\alpha^{(k)}$ ($0 \leq \alpha^{(k)} \leq 1$), where $\alpha^{(k)}$ should decrease as $k$ increase. Then we have
% \begin{flalign}
% {{\mu}^{i_{m} ,j_n}_r}^{(k+1)} = (1-\alpha^{(k)}){{\mu}^{i_{m} ,j_n}_r}^{(k)} + \alpha^{(k)}{{}{\tilde{\mu}^{i_{m} ,j_n}_r}}^{(k)} . 
% \label{eq_msa}
% \end{flalign}

Since we validate the convergence of $\beta$, the corollary of Banach fixed-point theorem can be used to develop the solution procedures. The detailed steps are summarized in Algorithm \ref{alg_overall}. To address the randomness in the network loading model, we define a "burn-in" iteration $K_b$ for $\beta$ selection, and a maximum iteration $K_t$ for algorithm termination (analogue to the Markov chain Monte Carlo (MCMC) methods). $\beta^{\text{ini}}$ is the initial value of $\beta$. 

\begin{algorithm} 
\caption{Solution Procedures for Path Choice Estimation} \label{alg_overall}
\begin{algorithmic}[1]
\State Initialize $\beta^{(0)}=\beta^{\text{ini}}$. 
\State Initialize $\boldsymbol{\mu}^{(0)} = \text{Network Loading }(\beta^{(0)}, \boldsymbol{q_e}, \theta)$ (sub-problem 2)  
\State Set iteration counter $k=0$.
\Do
    \State $k = k + 1$
    \State Solve sub-problem 1(a) with fixed $\boldsymbol{\mu}^{(k-1)}$ and return $\boldsymbol{p}^{(k)}$
    \State Solve sub-problem 1(b) with fixed $\boldsymbol{p}^{(k)}$ and return $\beta^{(k)}$
    \State Solve sub-problem 2 with $\beta^{(k)}$ as input and return $\boldsymbol{\mu}^{(k)}$
    % \State ${{\mu}^{i_{m} ,j_n}_r}^{(k)} = (1-\alpha^{(k)}){{\mu}^{i_{m} ,j_n}_r}^{(k-1)} + \alpha^{(k)}{{}{\tilde{\mu}^{i_{m} ,j_n}_r}}^{(k)}$
\doWhile{$k \leq K_t$}
\State $\beta = \sum_{k=K_b} ^ {K_t} \beta^{(k)}/(K_t-K_b+1)$
\State \Return $\beta$
\end{algorithmic}
\end{algorithm}

\section{Case study and Model validation}\label{case_study}
For the purpose of model illustration and validation, we apply the proposed modeling framework on Hong Kong MTR network. The model is validated using both synthetic data and real-world AFC data. 

\subsection{Hong Kong MTR Network}
The map for Hong Kong MTR system is shown in Fig \ref{fig_map}. In this study, the airport express and light rail transit services are not considered since they are separated from the urban railway lines and passengers who enter the urban railway lines from these services need to tap-in again. The system consists of 10 lines and 114 stations, where 16 out of them are transfer stations. In this network, most transfer stations connect only two lines. A special case is Admiralty station in the Hong Kong island, where three lines pass through the same transfer station. The Admiralty station is in the CBD area of Hong Kong. So during peak hour there are many passengers boarding in these stations and head to the north of the city. 
\begin{figure}[H]
\centering
\includegraphics[width=6 in]{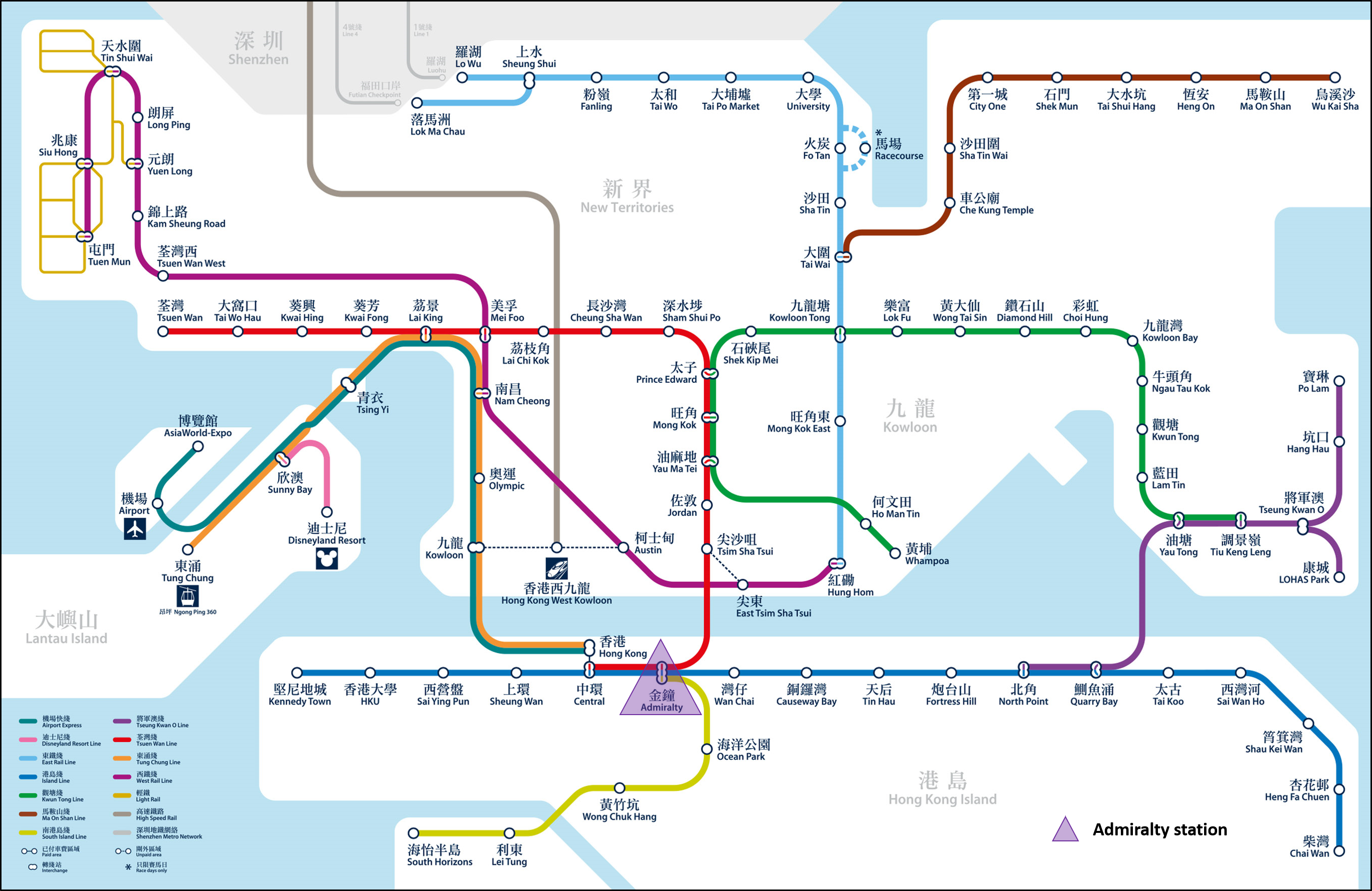}
\caption{Hong Kong MTR Metro System Map}
\label{fig_map}
\end{figure}

\subsection{Validation setting}
We use the AFC data on March 16th (Thursday), 2017 for the model validation. The path sets for each OD pair are provided by MTR. \cite{Li2014MTRChoice} has conducted a revealed-preference (RP) route choice survey of more than 20,000 passengers in the MTR system. The estimation results are shown in Appendix \ref{append_weixuan_results}. According to \cite{Li2014MTRChoice}, the following attributes were used to quantify path utility: (a) total in-vehicle time, (b) number of transfer times, (c) relative walking time (total walking time divided by total route distance) and (d) the commonality factor (Eq. \ref{eq_CF}). The evening peak (18:00-19:00) is selected for validation. For simplicity, we assume the path shares are static during this hour. We set the weights in the objective function of sub-problem 1(a) as $w_1 = 1$ and $w_2 = 0$, which means no prior knowledge is available. The maximum iteration is $K_t$ is set as 15 and "burn-in" iteration $K_b$ is set as 13. $\beta^{\text{ini}}$ is set as 0 for all parameters. The system parameters $\theta$ (constants) of network loading model is summarized below. Access/Egress walking time is defined as the walking time between fare machine and train boarding platform. Warm up (cool down) time indicates the time before (after) simulation period start (end). It is needed because simulation system usually start from empty (no train and passengers). But real world does not start without train/passengers in-progress. 
\begin{itemize}
\item Access/Egress walking time: Platform-specific, collected from MTR field measurement.
\item Transfer walking time: Platform-specific, collected from MTR field measurement.
\item Time table: Obtained from MTR operation team. Future research can use AVL data to get real-world train arrival and departure information.
\item Capacity: 230 passengers per car according to MTR congestion standard. The number of cars for each train is obtained from MTR operation team.
\item Warm up and Cool down time: 60 minutes warm up and cool down time suggested by \cite{Ma2019Network}.
\end{itemize}

Since the real-world path choice information usually unavailable, it is common to quantitatively validate model with synthetic data. To generate the synthetic data, we first extract the OD entry flow from the real-world AFC records. Choice parameters $\beta$ estimated in \cite{Li2014MTRChoice} are treated as people's "true" behavior parameters (called synthetic $\beta$ hereafter). Then, we use the network loading model with the true OD entry flow and the synthetic $\beta$ as input to simulate the travel of passengers in the system, and record people's tap-in and tap-out time. The records of all people's tap-in and tap-out time are treated as \emph{synthetic AFC data}. For model validation, we can apply the proposed model to the synthetic AFC data and compare the estimated $\beta$ with the synthetic $\beta$. However, in terms of real-world validation, as the ground-truth path shares are unavailable, some qualitatively analysis and indirect comparison are conducted. Details can be found in Section \ref{realworld}.

\subsection{Benchmark Model}
To compare the model performance, we use a purely simulation-based optimization (SBO) method \citep{Mo2019Calibrating} as the benchmark. The formulation is shown below.
\begin{mini!}|s|[2]                   % mini! = minimize 
    {\beta}                               % optimization variable
    {w_1\sum_{i_m,j_n}(q^{i_{m}  ,j_{n}}-\tilde{q}^{i_{m} ,j_{n} })^2 + w_2||\beta - \tilde{\beta}||^2 \label{eq:obj_sbo}}   % objective function and label
    {\label{eq:sbo}}             % label for optimizatio problem
    {}                                % optimization result
    \addConstraint{{q}^{i_{m} ,j_n}}{ = \text{Network Loading }(\beta, q^{i_m,j}, \theta)}{ \quad  \quad \forall i_m, j_n}
    \addConstraint{L_{\beta} \leq \beta \leq U_{\beta}}{}{}    
\end{mini!}
where $L_{\beta}$ and $U_{\beta}$ are the pre-determined lower and upper bound of $\beta$. We set $w_1 = 1$ and $w_2 = 0$ as above. $\theta$ is the parameter of the network loading model, which is set equal as above. Compared with our proposed model, the purely SBO method is closer to brute-force searching. So $L_{\beta}$ and $U_{\beta}$ are usually required to narrow the feasible space and make the algorithm work. The values of $L_{\beta}$ and $U_{\beta}$ are shown in table \ref{tab_beta_compare}. By introducing $L_{\beta}$ and $U_{\beta}$, we actually provide the benchmark model with more information. 

Many solution algorithms have been proposed to solve SBO problems. Three major classes includes direct search method, gradient-based method, and the response surface method \citep{osorio2013simulation,amaran2016simulation}. According to \cite{osorio2013simulation} and \cite{cheng2019surrogate}, the response surface method is recently more popular in transportation domain and presents better performance. Thus, in this study, we adopt two response surface methods to solve the benchmark model: bayesian optimization (BYO) \citep{snoek2012practical} and constrained optimization using response surfaces (CORS) \citep{regis2005constrained}. BYO aims to constructs a probabilistic model for the objective function (response surface) and then exploits this model to determine where to evaluate the objective function for next step. In each iteration the probabilistic model will be updated according to the posterior distribution. CORS also need to construct a response surface model, and update the model based on all previously probed points at each iteration. The principles for next evaluated points selection are: (a) finding new points that have lower objective function value, and (b) improving the fitting of response surface model by sampling feasible regions where little information exists. For more details regarding these two methods people can refer to \cite{snoek2012practical} and \cite{regis2005constrained}. 

SBO methods are usually unstable due to the randomness in searching process. We therefore run each solution algorithm for 10 replications and show the mean and standard deviation of objective function curves. For comparison purpose, the SBO methods are only conducted in the synthetic data set because we can only compare the accuracy of path choice estimation in the synthetic data.

\subsection{Results of Synthetic Data}
Two indicators are reported during the iteration. One is the objective function, another is the root-mean-square-error (RMSE). The formula of RMSE is shown below.
\begin{flalign}
\label{eq_RMSE}
\text{RMSE} = \sqrt{\sum_{i_m,j}\sum_{r \in \mathscr{R}(i,j)}(p^{i_m ,j}_r - \hat{p}^{i_m ,j}_r)^2/\sum_{i,j}R_{i,j}}
\end{flalign}
where $p^{i_m ,j}_r$ are the estimated path shares and $\hat{p}^{i_m ,j}_r$ are the synthetic path shares (unit is \%). $\sum_{i,j}R_{i,j}$ is the total number of paths in the system. 

The curve of objective function is shown in Figure \ref{fig_obj}. The error bars for benchmark methods represent the standard deviation. We found the our proposed metthod can dominate the benchmark models both in convergence speed and in final convergent results. The RMSE comparison results are shown in Figure \ref{fig_rmse}. We can also observe our method can approach the "true" path shares rapidly, and obtain lower estimation error than the benchmark models. Note that the RMSE may not always decrease with the reduction of objective function. This is because the relationship between path choices and OD entry-exit flows is highly non-linear. 

\begin{figure}[H]
\centering
\includegraphics[width=5 in]{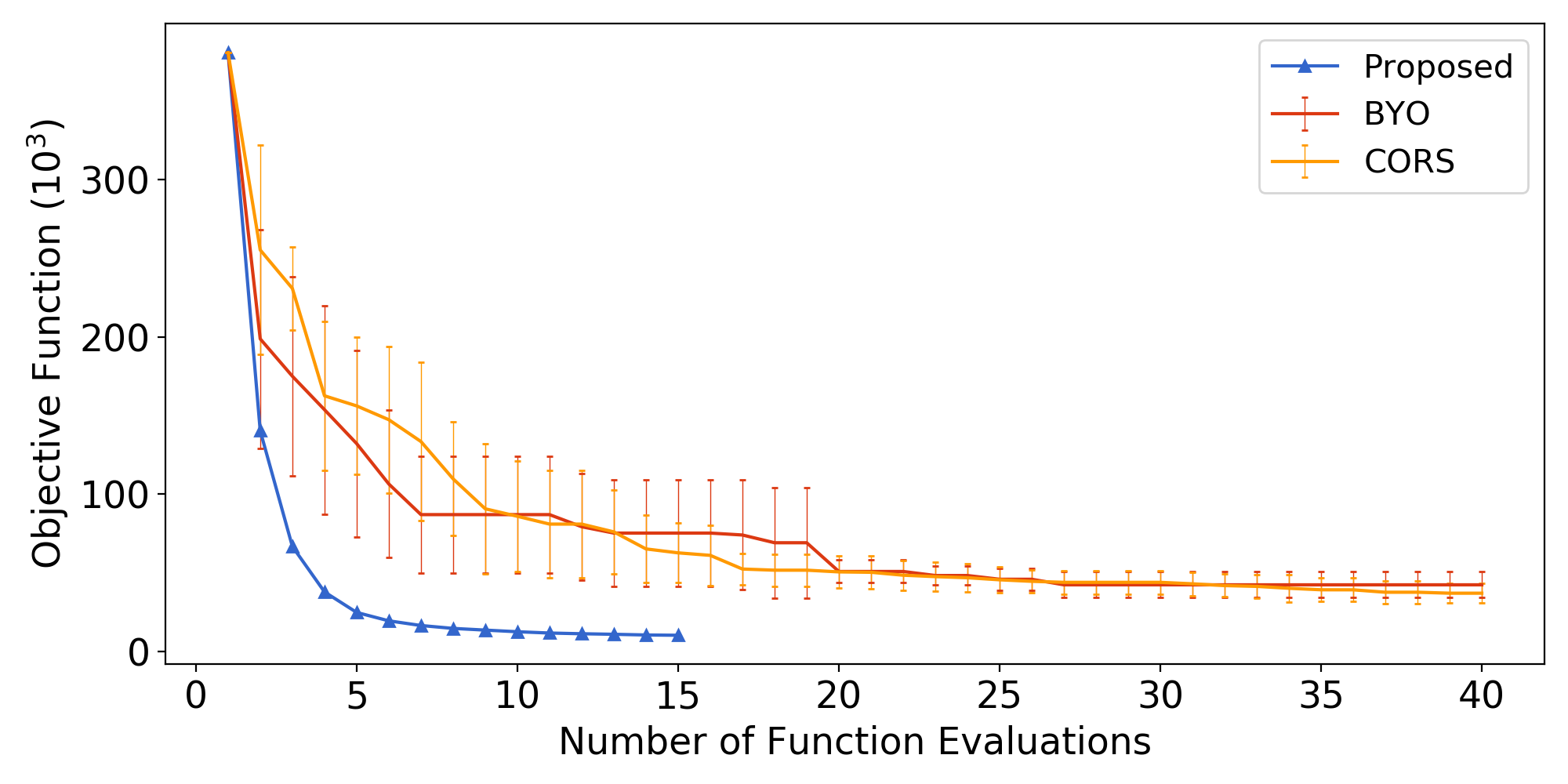}
\caption{Objective Function Results of Synthetic Data}
\label{fig_obj}
\end{figure}

\begin{figure}[H]
\centering
\includegraphics[width=5 in]{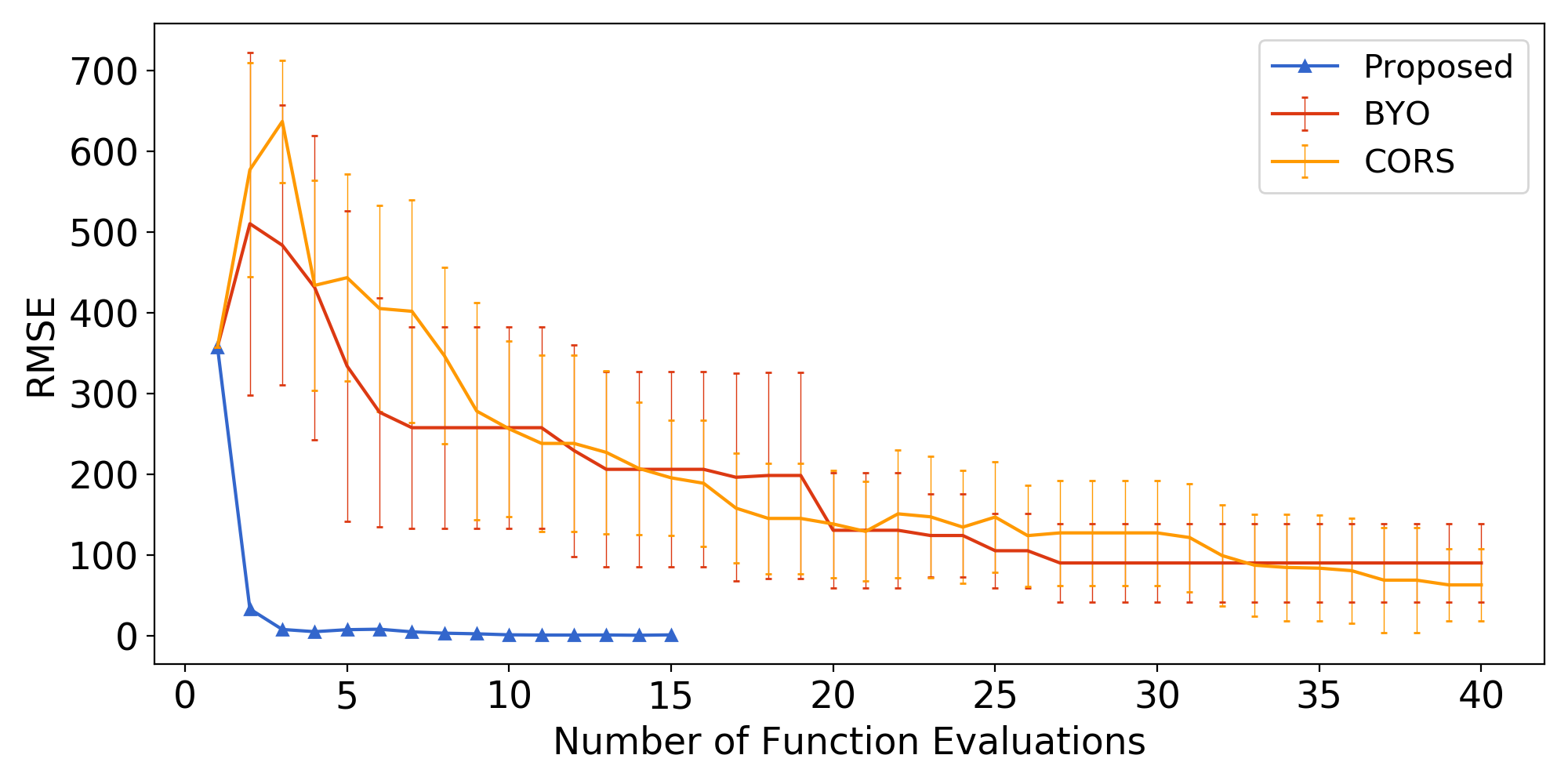}
\caption{RMSE Results of Synthetic Data}
\label{fig_rmse}
\end{figure}

The comparison of estimated $\beta$ and synthetic $\beta$ are shown in Table \ref{tab_beta_compare}. The estimated $\beta$ of the proposed method are very close to the synthetic ones, and also outperforms the $\beta$ estimated from the benchmark models.

\begin{table}[H]
\caption{$\beta$ Estimation Results of Synthetic Data}
\centering
\begin{tabular}{@{}llllll@{}}
\toprule
\multirow{2}{*}{Variable} & \multirow{2}{*}{Synthetic} & \multicolumn{3}{c}{Estimated} & \multicolumn{1}{c}{\multirow{2}{*}{{[}$L_{\beta}$,$U_{\beta}${]}}} \\ \cmidrule(lr){3-5}
                          &                            & Proposed  & BYO     & CORS    & \multicolumn{1}{c}{}                                               \\ \midrule
In-vehicle time ($\beta_1$)          & -0.147                     & -0.156    & -0.205  & -0.231  & {[}-2, 0{]}                                                        \\
Number of transfers ($\beta_2$)      & -0.573                     & -0.544    & -1.218  & -1.189  & {[}-4, 0{]}                                                        \\
Relative walking time  ($\beta_3$)    & -1.271                     & -1.291    & -2.499  & -2.316  & {[}-6, 0{]}                                                        \\
Commonality factor ($\beta_4$)       & -3.679                     & -3.413    & -6.184  & -6.537  & {[}-10, 0{]}                                                       \\ \midrule
Objective function        & -                          & 10328.8   & 42390.6 & 37066.1 & -                                                                  \\
RMSE                      & -                          & 1.36     & 90.24   & 63.09   & -                                                                  \\ \bottomrule
\end{tabular}
\label{tab_beta_compare}
\end{table}

\subsection{Results of Real-world Data}\label{realworld}
In terms of real-world data, we cannot conduct the directly comparison since the ground truth path choice is unavailable. We first conduct a qualitative analysis in terms of estimated parameters (see Table \ref{tab_real_beta}). Compared with the $\beta$ obtained by \cite{Li2014MTRChoice} (i.e. the synthetic $\beta$ in Table \ref{tab_beta_compare}), the scale of all coefficients are similar. The trade off between in-vehicle time and number of transfers are reasonable, where one transfer is equivalent to 7.9 minutes of in-vehicle travel time. The trade-off between in-vehicle time and walking time is relatively small for long trips but significant for short trips. The results indicate that for a trip with 4 
stations (around 5 cm of map distance), one minute of transfer walking time is equivalent to 2.51 minutes of in-vehicle travel time. For a trip with 8 stations, (around 12 cm in map distance), one minute of walking time is equivalent to 1.05 minutes of in-vehicle travel time. The substitution patterns are reasonable and similar to the previous results \citep{Li2014MTRChoice}.

\begin{table}[H]
\caption{$\beta$ Estimation Results of Real-world Data}
\centering
\begin{tabular}{@{}lllll@{}}
\toprule
In-vehicle time & Number of transfers & Relative walking time & Commonality factor \\ \midrule
-0.116        & -0.920              & -1.457              & -1.775            \\ \bottomrule
\end{tabular}
\label{tab_real_beta}
\end{table}

Though we cannot directly compare path shares, some other indicators (e.g. left behind rate) can also reflect the quality of path shares. We have the field observation data for Admiralty station Northbound platform during the testing period (18:00-19:00). The data is collected by MTR employees who counted the passengers in the platform. The average left behind rate, total number of arrival passengers (sum of tap-in and transfer passengers) and total number of boarding passengers during 18:00-19:00 are recorded. These indicators can also be obtained from network loading model which takes path shares as input. So we can input the estimated path shares into the network loading model and compare the output indicators with the ground truth. Two other path shares are used to compare with the estimated path share. One is a naive path share which assume all paths are equal likely to be chosen (named as "uniform" in Table \ref{tab_real_compare}). Another is the path share calculated from \cite{Li2014MTRChoice}, which is currently used by MTR agency for operation purpose. 

The comparison results are shown in Table \ref{tab_real_compare}. Compared with the ground truth, the estimated path shares can generate very close number of arrival passengers and number of boarding passengers. The square error of OD entry-exit flow (i.e. the objective function) is also the lowest. Although the output left behind rate (LBR) is not as good as two baseline path shares, we observe this indicator cannot completely reflect the quality of path shares because the uniform path choice tends to generate the best LBR. But from the other three indicators we know uniform is the worst path choice among the three. Since the objective function is a network-level indicator but the other three are local platform-level, the estimated path choice is the most balanced one and gives more network-level goodness of fit, which is more suitable for the system evaluation.

\begin{table}[H]
\caption{Indicators comparison of Admiralty station (18:00 to 19:00)}
\centering
\resizebox{1\textwidth}{!}{
\begin{tabular}{@{}lcccc@{}}
\toprule
\multirow{2}{*}{} & \multicolumn{4}{c}{Indicators}                                                                                                                                                                                                                                              \\ \cmidrule(l){2-5} 
                  & Left behind rate & \begin{tabular}[c]{@{}c@{}}Number of \\ arrival passengers\end{tabular} & \begin{tabular}[c]{@{}c@{}}Number of \\ boarding passengers\end{tabular} & \begin{tabular}[c]{@{}c@{}}Objective function \\ (Square error of OD entry-exit flow)\end{tabular} \\ \midrule
Ground-truth      & 0.767            & 24,945                                                                   & 24,696                                                                    & -                                                                                              \\
Proposed model    & 0.705            & 24,890                                                                  & 24,403                                                                   & 1,044,692                                                                                              \\
Li (2014)         & 0.734            & 24,959                                                                   & 23,125                                                                    & 1,170,160                                                                                              \\
Uniform           & 0.779            & 25,683                                                                   & 19,599                                                                    & 1,289,672                                                                                             \\ \bottomrule
\end{tabular}
}
\label{tab_real_compare}
\end{table}

\subsection{Robustness Testing}
Model robustness is an important indicator for the real-world application. In this section we test two perspectives: different initial $\beta$ and different case study dates. A robust model should output similar estimated $\beta$ regardless of initial values. In terms of different case study dates, the estimated $\beta$ for all weekdays in the same week should be similar since passengers' choice behaviors are stable during a short term.

\subsubsection{Different Initial $\beta$}
The tests of different initial $\beta$ are conducted in the synthetic data set because we can compare the distribution of estimated $\beta$ and synthetic $\beta$ in this way. The initial $\beta$ are drawn from the uniform distribution $\text{U}(L_\beta,U_\beta)$ for 12 replications. Figure \ref{fig_robust_obj} shows the convergence of objective function with respect to different initial $\beta$. The initial objective function varies a lot given different initial $\beta$. But after around 10 iterations. All objective function curves converge to the same value, which demonstrates the robustness of model with respect to initial values. Figure \ref{fig_robust_boxplot} is the boxlpot of estimated $\beta$ of different replications. The name of $\beta_1$,...,$\beta_4$ can be found in Table \ref{tab_beta_compare}. The estimated value of $\beta_1$, $\beta_2$ and $\beta_3$ are very stable regardless of initial values. While the estimated $\beta_4$ (i.e. commonality factor) shows some fluctuations, but still within a small range (95\% confidence interval is around $-3.2\sim-3.6$). This is corresponding to the survey estimation results where $\beta_4$ has a relatively low t-value (see \ref{append_weixuan_results}). It is worth noting that similar to the results in Table \ref{tab_beta_compare}, we cannot perfectly estimate the $\beta_4$, which leads to the synthetic $\beta_4$ located outside the 95\% confidence interval. The reason may be the collinearity between commonality factor and other variables. But compared with the baseline models, our method already output a good estimate.

\begin{figure}[H]
\centering
\subfloat[Convergence of Objective Function (Different curves indicate different initial $\beta$)]{\includegraphics[width=0.4\textwidth]{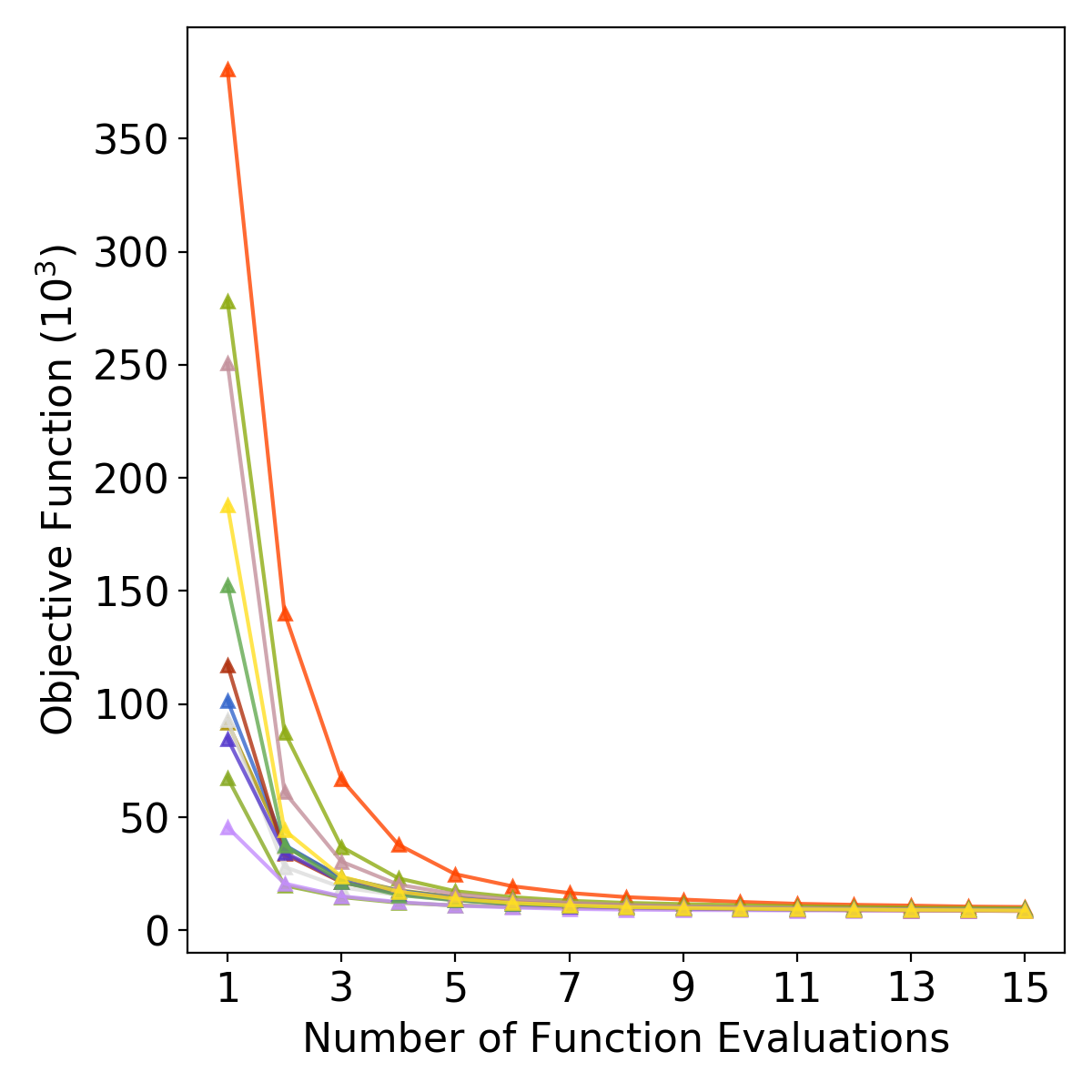}\label{fig_robust_obj}}
\hfil
\subfloat[Boxplot of Estimated Coefficients with Different Initial $\beta$]{\includegraphics[width=0.4\textwidth]{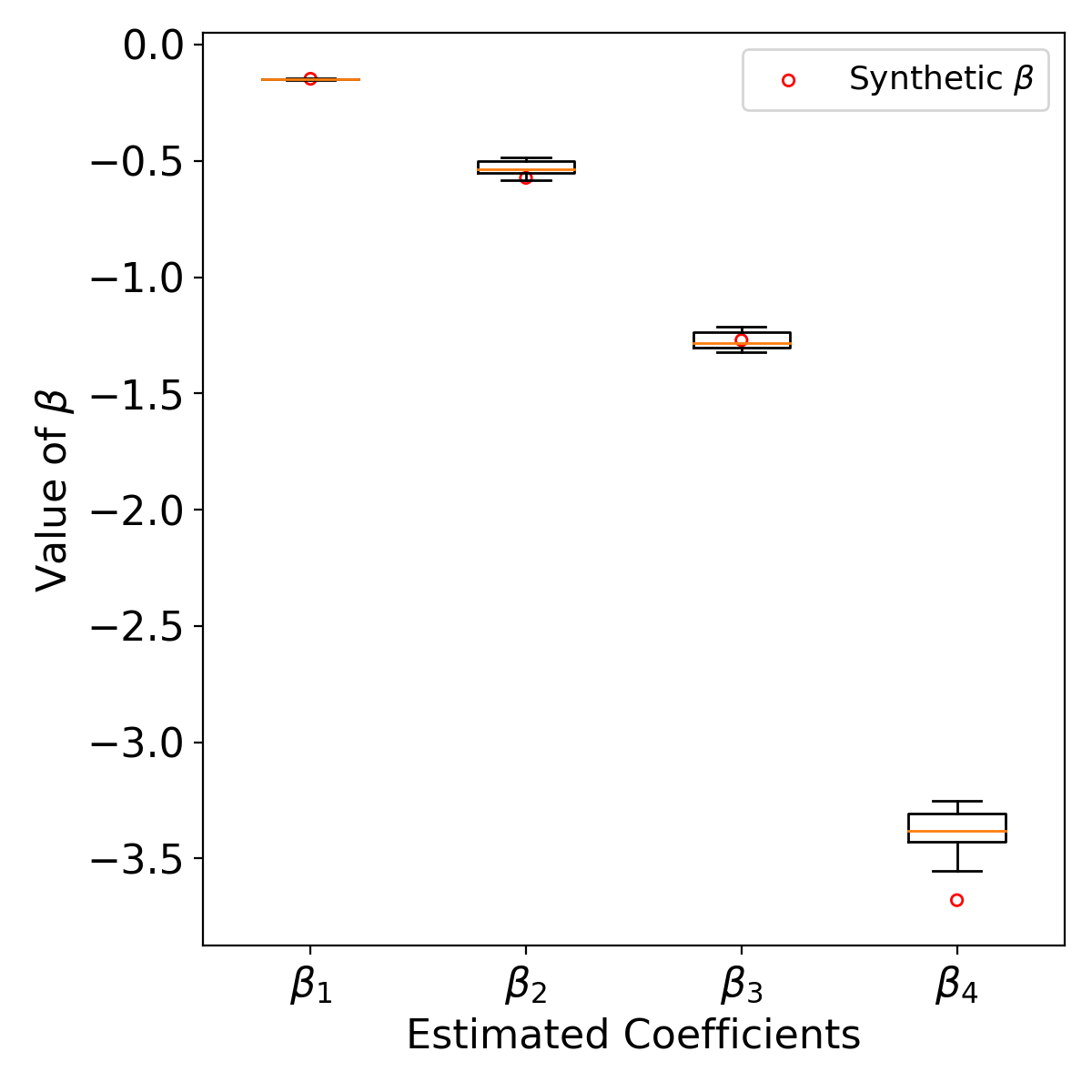}\label{fig_robust_boxplot}}
\hfil
\label{fig_beta_compare}
\caption{Testing Results of Different Initial $\beta$ using Synthetic Data}
\end{figure}

\subsubsection{Different Case Study Dates}
To test the robustness in terms of different case study dates, we applied our model to the real-world data in the week from March 13rd to March 17th, 2017. The estimated $\beta$ comparison is shown in Figure \ref{fig_robustness_dates}. In general, all estimated values are stable regardless of case study dates except for the coefficients of relative walking time on Friday. This may be due to that Friday night is the start of weekends, in which people have more entertainment trips in evening peak hour. And the walking time is less sensitive for entertainment trips comparing with the commuting trips. Overall, the proposed model is robust in terms of different case study weekdays in a short-term, which implies the underlying path choice behaviors of passengers is captured by the model. 

\begin{figure}[H]
\centering
\includegraphics[width=5 in]{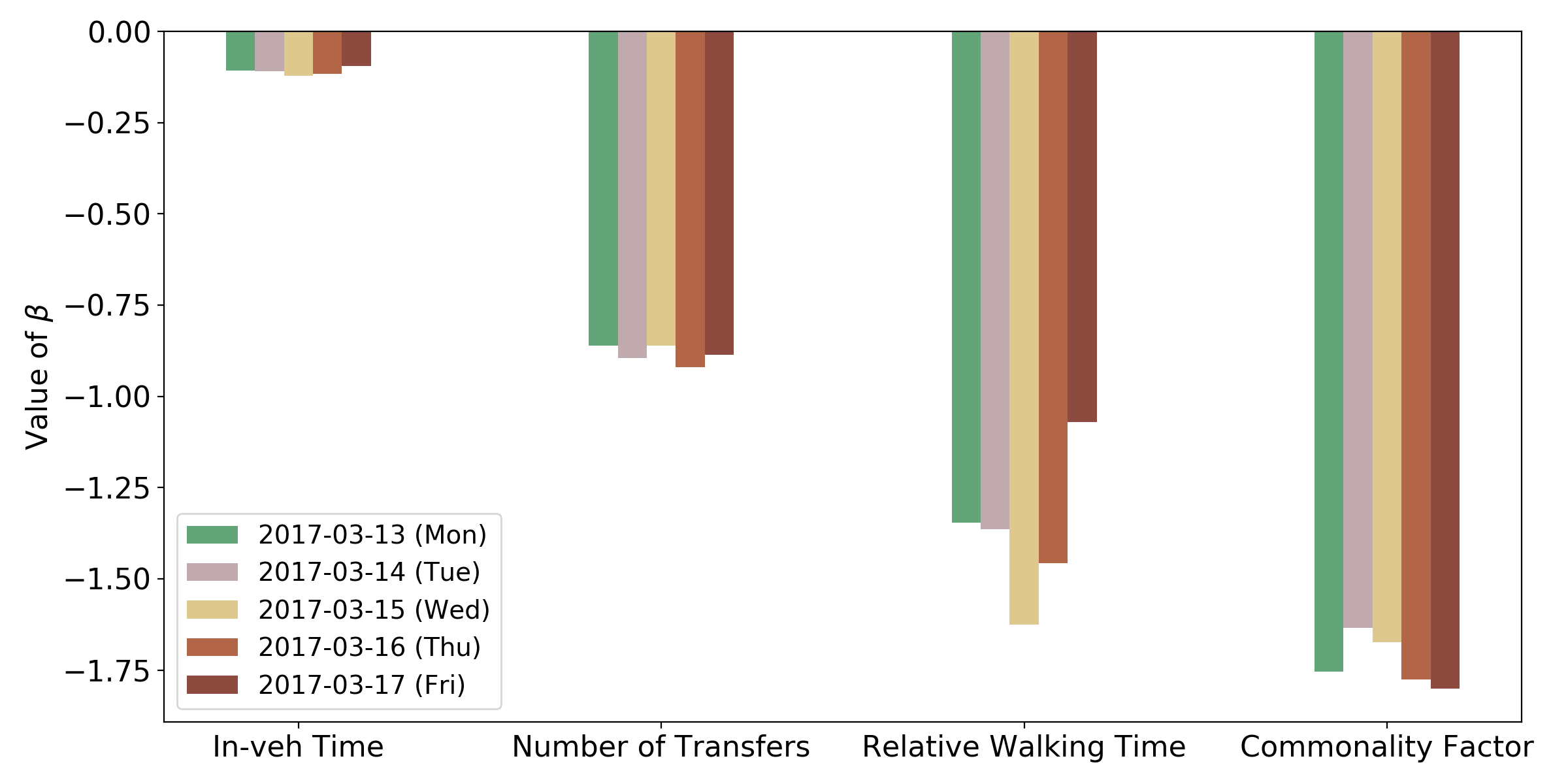}
\caption{Estimated $\beta$ Comparison of Different Dates using Real-world Data}
\label{fig_robustness_dates}
\end{figure}
\section{Conclusion and Discussion}\label{discuss}
In this paper, we developed an assignment-based approach to infer the passenger route choice behavior. It models the correlation of crowding among platforms and interactions between path choice and crowding explicitly. The path choice estimation is modeled as a optimization problem. The original intractable problem is decomposed into three tractable sub-problems which can be solved efficiently. Case studies using synthetic and actual data in Hong Kong MTR system validate the effectiveness and robustness of the approach.

The advantage of this framework lies in embedding the network assignment into path choice estimation, which can incorporate the interaction between left behind and path choices. The model can be generalized to accommodate different choice model structures. For example, people can leverage path-size logit and cross nested logit model to describe to choice behaviors. The revision of the framework only requires the change of sub-problem 1(b). Moreover, more path attributes, such as waiting time, expected left-behind rate, can also be added into the consideration.

The developments in this paper have been focused on a general framework, while the models and examples we presented in this paper still have some limitations. First, we imposed the assumption on route behavior modeling that only one set of $\beta$ is applied for the whole network. The real-world route choice behaviour may be more diverse and heterogeneous. Future research can cluster different OD pairs with different $\beta$ based on the corresponding passengers' characteristics. The model can be easily generalized with multiple sets of $\beta$. Second, we assume a fixed physical capacity for the train in network loading. In the reality, the number of people can board a train may vary from stations and crowding levels \citep{Ma2019Network, liu2016willingness}, which may also affect the path choice estimation. Future studies can incorporate the concepts of willingness to board (WtB) into the model and co-estimate the path choice, left behind and WtB.

\section{Author Statement}
The authors confirm contribution to the paper as follows: study conception and design: B. Mo, Z. Ma, H.N. Koutsopoulos, J. Zhao; data collection: B. Mo, Z. Ma; analysis and interpretation of results: B. Mo, Z. Ma; draft manuscript preparation: B. Mo. All authors reviewed the results and approved the final version of the manuscript.

\section{Acknowledgements}
The authors would like to thank Hong Kong Mass Transit Railway (MTR) for their support and data availability for this research. 

\section*{References}

\bibliography{mybibfile}

\appendix
\appendixpage
\section{Passenger Route Choice Model for MTR System} \label{append_weixuan_results}
These results are from \cite{Li2014MTRChoice}. The C-logit Model formulation is same to Eq. (\ref{eq_MNL}) and Eq. (\ref{eq_CF}). A total number of 31,640 passengers completed the questionnaire. After filtering duplicate responses, 26,996 responses were available. The model results are shown in Table \ref{tab_choice}. The main explanatory variables are the total in-vehicle time, relative transfer walking time and number of transfers. All variables are statistically significant with the expected signs. Routes with high in-vehicle time, walking time and number of transfers are less likely to be chosen by passengers.

\begin{table}[H]
\caption{Route Choice Model Estimation Results}
\centering
\begin{tabular}{@{}lllll@{}}
\toprule
                               & Estimate & Std. Error & t-value &     \\ \midrule
In-vehicle time                & -0.147   & 0.011      & -13.64  & *** \\
Relative walking time & -1.271   & 0.278      & -4.56  & *** \\
Number of transfers            & -0.573   & 0.084      & -6.18   & *** \\
Commonality factor             & -3.679   & 1.273      & -2.89   & **  \\ \midrule
\multicolumn{5}{l}{$\rho^2= 0.54$ }                                  \\ \bottomrule
\multicolumn{5}{l}{***: $p < 0.01$; **: $p<0.05$.} \\
\end{tabular}
\label{tab_choice}
\end{table}
\end{document}